\documentclass{article}
\usepackage{amssymb,amsfonts,amsmath,amsthm,amscd,dsfont,mathrsfs}
\usepackage{graphicx,graphics,epsfig,psfrag}
\usepackage[usenames,dvipsnames]{color,colortbl}
\usepackage{url,balance,cite}
\usepackage{array,multirow}
\usepackage{verbatim}
\usepackage{subfig}
\usepackage{xspace}

\usepackage{graphicx} 

\usepackage{algorithm}
\usepackage{algorithmic}

\usepackage{hyperref}

\newcommand{\ignore}[1]{}
\newcommand{\etal}{\textit{et al.}}
\newcommand{\eg}{e.g.}

\newcommand{\ie}{i.e.}

\newcommand{\ml}{Movielens\xspace}
\newcommand{\fl}{Flixster\xspace}

\newcommand{\techreport}[2]{#2}

\newcounter{packednmbr}
\newenvironment{packedenumerate}{\begin{list}{\thepackednmbr.}{\usecounter{packednmbr}\setlength{\itemsep}{0pt}\addtolength{\labelwidth}{-5pt}\setlength{\leftmargin}{\labelwidth}\setlength{\listparindent}{\parindent}\setlength{\parsep}{0pt}\setlength{\topsep}{3pt}}}{\end{list}}
\newenvironment{packeditemize}{\begin{list}{$\bullet$}{\setlength{\itemsep}{0pt}\addtolength{\labelwidth}{-5pt}\setlength{\leftmargin}{\labelwidth}\setlength{\listparindent}{\parindent}\setlength{\parsep}{0pt}\setlength{\topsep}{3pt}}}{\end{list}}
\newenvironment{packeddescription}[1][0pt]
  {\list{}{\labelwidth=0pt \leftmargin=#1 \setlength{\parsep}{0pt}\setlength{\topsep}{3pt}
   }}
  {\endlist}

\makeatletter{}
\def\<{\langle}
\def\>{\rangle}
\def\reals{{\mathbb R}}
\def\naturals{{\mathbb N}}
\def\cV{{\cal V}}

\def\cY{{\cal Y}}

\def\cX{{\cal X}}
\def\cL{{\cal L}}
\def\Pm{{\rm P}}
\def\Qm{{\rm Q}}
\def\Em{{\rm E}}
\def\ed{\stackrel{{\rm d}}{=}}
\def\hx{\widehat{x}}
\def\hr{\widehat{r}}
\def\bhx{\widehat{\bf x}}
\def\cR{{\cal R}}
\def\sT{{\sf T}}
\def\id{{\rm I}}

\newtheorem{theorem}{Theorem}\newtheorem{definition}{Definition}

\newtheorem{lemma}{Lemma}

\definecolor{red}{RGB}{255,0,0}

\definecolor{dred}{RGB}{210,20,50}

\definecolor{green}{RGB}{0,255,0}

\definecolor{dgreen}{RGB}{0,151,0}

\definecolor{lgreen}{RGB}{75,180,130}

\definecolor{blue}{RGB}{0,0,255}

\definecolor{lblue}{RGB}{30,160,240}

\definecolor{magenta}{RGB}{255,0,255}

\definecolor{orange}{RGB}{255,128,0}

\definecolor{dorange}{RGB}{255,165,0}

\definecolor{violet}{RGB}{207,74,221}

\definecolor{grey}{RGB}{191,191,191}

 \def\MSE{{\rm MSE}}

\newcommand{\eps}{\varepsilon}

\newcommand{\junk}[1]{}

\def\b0{{\bf 0}}

\def\ba{{\bf a}}
\def\bb{{\bf b}}

\def\bv{{\bf v}}

\def\bx{{\bf x}}

\newcommand{\leak}{\ensuremath{L}}
\newcommand{\obfs}{\ensuremath{Y}}
\newcommand{\argmin}{\ensuremath{\mathop{\arg\min}}}

\newcommand{\myset}[1]{\mathcal{#1}}

\begin{document}

\title{Privacy Tradeoffs in Predictive Analytics}

\author{Stratis Ioannidis$^*$, Andrea Montanari$^\dagger$, Udi Weinsberg$^*$,\\Smriti Bhagat$^*$, Nadia Fawaz$^*$, Nina Taft$^*$\\$^*$ Technicolor, $^\dagger$ Stanford University}
\maketitle
\begin{abstract}
\makeatletter{}Online services routinely mine user data to predict user preferences, make recommendations, and place  targeted ads. Recent research has demonstrated that several private user attributes (such as political affiliation, sexual orientation, and gender) can be inferred from such data. Can a privacy-conscious user  benefit from personalization while simultaneously protecting her private attributes? We study this question in the context of a rating prediction service based on matrix factorization. 
We construct a protocol of interactions between the service and users that has remarkable optimality properties: it is \emph{privacy-preserving}, in that no inference algorithm can succeed in inferring a user's private attribute with a probability better than random guessing; it has \emph{maximal accuracy}, in that no other privacy-preserving protocol improves rating prediction; and, finally, it involves a \emph{minimal disclosure}, as  the prediction accuracy strictly decreases when the service reveals less information. We extensively evaluate our protocol using several rating datasets, demonstrating that it successfully blocks the inference of gender, age and political affiliation, while incurring less than 5\% decrease in the accuracy of rating prediction.
 
\end{abstract}

\makeatletter{}
\section{Introduction}

Online users are routinely asked to provide feedback about their
preferences and tastes. Often, users give five-star ratings for movies, books, restaurants, or items they purchase, and ``like''  news articles, blog posts, pictures or other kinds of micro-content. Online services mine such feedback to predict users' future preferences, using techniques such as matrix factorization~\cite{candes2009exact,keshavan2010matrix,Koren:2009,Koren:2008}.
Such prediction can be employed to, e.g.,  make relevant product recommendations, to display targeted ads, or, more generally, personalize services offered; making accurate predictions is thus of fundamental importance to many online services.

Although users may willingly reveal, e.g., ratings to movies or ``likes'' to news articles and posts, there is an inherent privacy threat in this revelation.  To see this, consider the following general setting. An entity, which we call for concreteness the \emph{analyst},  has  a dataset of ratings given by users to a set of items (e.g., movies).  A private attribute of some users, such as their gender, age, political affiliation, etc., is also in the dataset. The analyst uses this dataset to offer a \emph{recommendation service}. Specifically, the analyst solicits ratings from new users; using these ratings, it predicts how these users would rate other items in its dataset (e.g., via matrix factorization techniques), and recommends items they are likely to rate highly. New users are \emph{privacy-conscious}: they want to receive relevant recommendations but do not want the analyst to learn their private attribute. However, having access to the above dataset, the analyst can potentially \emph{infer} the private attribute from the ratings they reveal.

The success of such inference clearly depends on how a user's feedback (i.e., her ratings) relates to her private attribute, and whether this correlation is evident in the dataset. 
Recent studies report many examples where strong correlations have been found: attributes successfully inferred from ratings or ``likes''  include political affiliation~\cite{kosinski2013private,salman:2013}, sexual orientation~\cite{kosinski2013private},   age~\cite{blurme:2012},  gender~\cite{blurme:2012,salman:2013}, and even drug use~\cite{kosinski2013private}.  Yet more privacy threats have been extensively documented in literature (see, e.g., \cite{Bhagat:2007,Mislove:2010,Otterbacher:2010,Rao:2010,Narayanan:2008}). It is therefore natural to ask \emph{how can a privacy-conscious user benefit from relevant recommendations, while preventing the inference of her private information}? Allowing this to happen is clearly desirable from the user's point of view. It also benefits  the analyst, as it incentivizes privacy-conscious individuals to use the recommendation service.

A solution proposed by many recent research efforts is to allow a user to distort her ratings before revealing them to the analyst~\cite{vaidya2005privacy,kasiviswanathan2011can,duchi2013local,calmon-allerton2012}.
This leads to a well-known tradeoff between \emph{accuracy} and \emph{privacy}: greater distortion yields better privacy but also less accurate prediction (and, hence, poorer recommendations).
We introduce for the first time a third dimension to this tradeoff, namely the  {\em information the analyst discloses to the users.}

To understand the importance of this dimension, consider the following hypothetical scenario. The analyst gives the privacy-conscious user an implementation of its rating prediction algorithm, as well as any data it requires--including, potentially, the full dataset at the analyst's disposal. The user can then execute this algorithm locally, identifying, e.g., which movies or news articles are most relevant to her. This would provide perfect privacy (as the user reveals nothing to the analyst) as well as maximal accuracy (since the user's ratings are not distorted). Clearly, this is untenable from the analyst's perspective, both for practical reasons (e.g., efficiency or code maintenance) and for commercial reasons: the analyst may be charging a fee for its services, and exposing such information publicly diminishes any competitive edge it may have.

The above hypothetical scenario illustrates that \emph{both} privacy \emph{and} accuracy can be trivially attained  when \emph{no constraints are placed on the information disclosed by the analyst}. On the other hand, such constraints are natural and necessary when the analyst's algorithms and data are proprietary.
 A natural goal is thus to determine the \emph{minimal} information the analyst needs to disclose to a privacy-conscious user, to enable a recommendation service that is both private and accurate. 
We make the following contributions:

\begin{packeditemize}
\item We introduce a novel mathematical framework to study issues of privacy, accuracy, and information disclosure when the analyst predicts ratings through matrix factorization (Section~\ref{sec:model}). In particular, we define a broad class of \emph{learning protocols} determining  the interactions between the analyst and a privacy-conscious user.  Each protocol specifies what information the analyst reveals, how the user distorts her ratings, and how the analyst uses this obfuscated feedback for rating prediction. \item We propose a simple learning protocol, which we call the \emph{midpoint} protocol (MP), and prove  it has remarkable optimality properties (Section~\ref{sec:MidpointProtocol}). First, it provides \emph{perfect privacy} w.r.t.~the user's private attribute: no inference algorithm predicts it better than random guessing.
Second, it yields \emph{optimal accuracy}: there is no privacy-preserving protocol allowing rating prediction at higher accuracy than MP.  Finally, the protocol involves a \emph{minimal disclosure}: any privacy-preserving protocol that discloses less information than MP  necessarily has a strictly worse prediction accuracy.
\item We extend our solution to handle common situations that occur in practice (Section~\ref{sec:Extensions}). We deal with the case where the user can only rate  a subset of the items for which the analyst solicits feedback: we provide a variant of MP, termed MPSS, and also establish its optimality in this setting. 
{We also discuss how the analyst can select the set of items for which to solicit ratings, and how the user can repeatedly interact with the analyst.}
\item We evaluate our proposed protocols on three datasets, protecting attributes such as user gender, age and political affiliation (Section~\ref{sec:evaluation}). We show that MP and MPSS attain excellent privacy: a wide array of inference methods are rendered no better than blind guessing, with an area-under-the-curve (AUC) below 0.55. This privacy is achieved with negligible impact (2-5\%) on rating prediction accuracy.  \end{packeditemize}

To the best of our knowledge, we are the first to take into account the data disclosed by an analyst  in the above privacy-accuracy tradeoff, and to establish the optimality of a combined disclosure, obfuscation, and prediction scheme. Our proofs rely on the modeling assumption that is the cornerstone of matrix factorization techniques and hence validated by vast empirical evidence (namely, that the user-item ratings matrix is approximately low-rank). Moreover, the fact that our algorithms successfully block inference against a barrage of different classifiers, some non-linear, 
further establishes our assumption's validity over real-world data.

\makeatletter{}\section{Related Work}\label{sec:RelatedWork}

\noindent \textbf{Threats.} Inference threats from user data have been extensively documented by several recent studies.
Demographic information has been successfully inferred from blog posts~\cite{Bhagat:2007}, search queries~\cite{Bi-www2013}, reviews~\cite{Otterbacher:2010}, tweets~\cite{Rao:2010}, and the profiles of one's Facebook friends~\cite{Mislove:2010}.
 In an extreme case of inference, Narayanan~\etal~\cite{Narayanan:2008}  show that disclosure of movie ratings  can lead to full de-anonymization (through a linkage attack), thus enabling unique identification of users.
Closer to our setting, Kosinski \emph{et al}.~\cite{kosinski2013private} show that several personality traits, including political views, sexual orientation, and drug use can be accurately predicted from  Facebook ``likes'', while Weinsberg \emph{et al.}~\cite{blurme:2012}  show that gender can be inferred from movie ratings  with close to 80\% accuracy. Salamatian~\etal~\cite{salman:2013} also show that political views can be inferred with confidence above $71\%$ by using only a user's ratings to her 5 most-watched TV shows.

\smallskip\noindent \textbf{Privacy-Preserving Data Mining and Information-Theoretic Models.}
Distorting data prior to its release to an untrusted analyst has a long history in the context of privacy-preserving data mining (see, e.g., \cite{warner1965randomized,vaidya2005privacy}). Distortion vs.~estimation accuracy tradeoffs have been studied in the context of several statistical tasks, such as constructing decision trees \cite{agrawal2000privacy}, clustering \cite{oliveira2003privacy,banerjee2012price},  and parameter estimation \cite{duchi2013local}.
 The outcome of such tasks amounts to learning an aggregate property from the distorted data of a user population. In contrast, we focus on estimating accurately a user profile to be used in matrix factorization, while keeping private any attribute she deems sensitive.

Our setting is closely related to the following information-theoretic problem \cite{yamamoto1983,calmon-allerton2012}. Consider  two dependent random variables $X$ and $Y$, where $X$ is to be released publicly while $Y$ is to be kept secret.  To prevent inference of $Y$ from the release, one can apply a distortion $f(X)$ on $X$; the goal is then to find the minimal distortion  so that the mutual information between $f(X)$ and $Y$ is below a threshold. This problem was originally addressed in the asymptotic regime \cite{yamamoto1983,Sankar-Poor-IFStrans2013}, while a series of recent works study it in a non-asymptotic setting~\cite{calmon-allerton2012,Makhdoumi2013privacy,salman:2013,rebollo2010t}. Broadly speaking, our work can be cast in this framework by treating a user's ratings as $X$, her private feature as $Y$, and  employing a correlation structure between them as specified by matrix factorization (namely, \eqref{eq:LinearModel}). Our definition of privacy then corresponds to zero mutual information (i.e., ``perfect'' privacy), and our protocol involves a minimal rating distortion.

We  depart from these studies of privacy vs.~accuracy (both in information-theoretic as well as the privacy-preserving data mining settings), by investigating a third axis, namely, the information disclosed by the analyst. To the best of our knowledge, our work is the first to characterize the disclosure extent necessary to achieve an optimal privacy-accuracy trade-off, an aspect absent from the aforementioned works.

\smallskip\noindent\textbf{Trusted Analyst.}
A different threat model than the one we study here considers a trusted analyst that aggregates data from multiple users in the clear. The analyst performs a statistical operation over the data,  distorts the output of this operation, and releases it publicly. The privacy protection gained by the distortion is therefore towards a third party that accesses the distorted output.  
The most common approach to quantifying privacy guarantees in this setting is through \emph{$\epsilon$-differential privacy} \cite{Dwork-McSherry-2006,dwork}. The statistical operations studied under this setting are numerous, including social recommendations~\cite{korolova}, covariance computation~\cite{mcsherry}, statistical estimation~\cite{smith2011privacy,dwork2009differential}, classification~\cite{chaudhuri2011differentially,rubinstein2012learning}, and principal component analysis~\cite{chaudhuri2013near}, to name a few. We differ in considering an untrusted analyst, and enabling a privacy-conscious user to interact with an analyst performing matrix factorization, rather than learning aggregate statistics.

\makeatletter{}\section{Technical Background}\label{sec:background}
  In this section, we briefly review matrix factorization and the modeling assumptions that underlie it. We also highlight privacy challenges that arise from its application. 
\subsection{Matrix Factorization (MF)}\label{sec:mf}
Matrix factorization \cite{Koren:2009,candes2010power,keshavan2010matrix} addresses the following prediction problem. A data analyst has access to a dataset in which $N$ users rate subsets of $M$ possible items (e.g., movies, restaurants, news articles, etc.). For $[N] \equiv
\{1,\dotsc,N\}$ the set of users, and $[M]\equiv\{1,\dotsc,M\}$ the set of
items, we denote by $\myset{E}\subset [N]\times [M] $ the user-item
pairs with a rating in the dataset. For $(i,j)\in \myset{E}$,  
let $r_{ij}\in\reals$ be user $i$'s rating to
item~$j$.
Given  the dataset $\{(i,j,r_{i,j})\}_{(i,j)\in\myset{E}}$, the analyst wishes to
predict the ratings for user-item pairs $(i,j)\notin \myset{E}$. 

Matrix factorization attempts to solve this problem assuming that the $N\times M$ matrix comprising all ratings is \emph{approximately low-rank}. In particular, it is assumed that for some small
dimension $d\in \naturals$ there exist vectors 
$\bx_i, \bv_j\in\reals^d$,   termed the \emph{user} and \emph{item} \emph{profiles}, respectively, such that
\begin{align}\label{lowrank}
r_{ij} = \langle \bx_i,\bv_j\rangle 
 +\varepsilon_{ij}, \quad\text{for }i\in [N],j\in [M],
\end{align}
where the ``noise'' variables $\varepsilon_{ij}$ are zero mean, i.i.d.~random variables with finite variance, and
$\<\ba,\bb\>\equiv\sum_{k=1}^da_kb_k$
is the usual inner product in $\reals^d$.
Given the ratings $\{r_{ij},~(i,j)\in\myset{E}\}$, the user and item profiles are typically computed through the following least-squares estimation (LSE) \cite{Koren:2009}:
\begin{align}\label{regrmse}
\min_{\{\bx_i\}_{i\in[N]}, \{\bv_j\}_{j \in [M]}}
\textstyle  \sum_{(i,j)\in \myset{E}}(r_{ij}-\langle \bx_i,\bv_j\rangle)^2.\end{align}
Minimizing this square error is a natural objective. Moreover, when the noise variables in \eqref{lowrank} are Gaussian, \eqref{regrmse}  is equivalent  to maximum likelihood estimation of user and item profiles. 
 Note that, having solved \eqref{regrmse}, the analyst can predict the rating of user $i$ for item $j$ as:
\begin{equation}\label{prediction}
  \hat{r}_{ij}\equiv\langle \bx_i, \bv_j\rangle, \quad (i,j)\notin \mathcal{E}.
\end{equation}
where $\bx_i, \bv_j$ are the estimated profiles obtained from \eqref{regrmse}.

Unfortunately, the minimization \eqref{regrmse} is \emph{not} a convex optimization problem. Nevertheless, there exist algorithms that provably recover the correct user and item profiles, under appropriate assumptions \cite{candes2009exact,candes2010power,keshavan2010matrix}. Moreover, simple gradient descent or alternating least-squares techniques are  known to work very well in practice~\cite{Koren:2009}.

\subsection{Incorporating Biases}\label{sec:biases}
Beyond user ratings, the analyst often has additional ``contextual'' information about users in the dataset. For example, if users are not privacy-conscious, they may reveal features such as their gender, age or other demographic information along with their ratings. Such information is typically included in MF through \emph{biases} (see, e.g., \cite{Koren:2008,Koren:2009}). 

Suppose, for concreteness, that each user $i$ discloses a binary feature $x_{i0}\in\{-1,+1\}$, e.g., their gender or political affiliation. This information can be incorporated in MF by adapting the model \eqref{lowrank} as follows:
\begin{align}\label{withbias}
r_{ij} = \langle \bx_i,\bv_j\rangle  + x_{i0}v_{j0} +\varepsilon_{ij} = \langle x_i,v_j \rangle +\varepsilon_{ij} 
\end{align}
for all $i\in [N]$, $j\in [M]$, where $v_{j0}\in \reals$ is a type-dependent bias, and $x_i=(x_{i0},\bx)\in \reals^{d+1}$, $i\in[N]$,  $v_j = (v_{j0},\bv_j)\in \reals^{d+1}$, $j \in [M],$ are \emph{extended} user and item profiles, respectively. Under this modeling assumption, the analyst can estimate profiles and biases jointly by solving:
\begin{align}\label{regrmsewithbias}
\min_{\{\bx_i\}_{i\in[N]}, \{(v_{j0},\bv_j)\}_{j \in [M]}} 
  \textstyle \sum_{(i,j)\in \myset{E}}(r_{ij}-\langle x_i,v_j \rangle)^2  .\end{align}
 Note that this minimization can be seen as a special case of \eqref{regrmse}, in which  extended profiles have dimension $d+1$, and the first coordinate of $x_i$ is fixed to either $-1$ or $+1$ (depending on the user's binary feature $x_{i0}$).
In other words, the feature $x_{i0}$ can be treated as yet another feature of a user's profile, though it is \emph{explicit} (i.e., a priori known) rather than \emph{latent} (i.e., inferred through MF). Prediction can be performed again through
$\hat{r}_{ij} = \langle x_i,v_j \rangle$, for $(i,j)\notin \myset{E}$.

Intuitively, the biases 
 $v_{j0}$  gauge the impact
of the binary feature $x_{i0}$ on each user's ratings. Indeed, consider sampling a random user from a population, and let $x=(x_{0},\bx)$ be her profile, where $\bx$ comprises the features that are independent of $x_{0}$. Then, it is easy to check from \eqref{withbias} that her rating $r_j$ for item $j$ will be such that:
$$\Em\{r_{j}\mid x_0 = 1\}- \Em\{r_j\mid x_0= 0\} =2v_{j0},$$
where the expectation is over the noise in \eqref{withbias}, as well as the random sampling of the user.
Put differently, given access to ratings by users that are not privacy-conscious and have disclosed, e.g., their gender $x_0$, $v_{j0}$ corresponds to half the distance between the mean ratings for item $j$ among genders.

Additional  explicit binary features can be incorporated similarly, by adding one bias per feature in \eqref{regrmsewithbias} (see, e.g., \cite{Koren:2009}). Categorical features can also be added through binarization; for simplicity, we focus on a single binary feature, discussing multiple and categorical features in Section~\ref{appendix:categorical}.

\subsection{Prediction for Privacy-Conscious Users}\label{sec:newuser}
Consider a scenario in which the analyst has performed MF over a dataset of users explicitly revealing a binary feature, and has extracted  the extended profiles $v_j=(v_{j0},\bv_j)\in \reals^{d+1}$ for each item $j\in [M]$.  Suppose now a privacy-conscious user  joins the system and \emph{does not explicitly reveal her private binary feature $x_0$ to the analyst}.

 In such a ``cold-start'' situation, the analyst would typically solicit a batch of ratings $\{r_j\}_{j\in S}$ for some set $\mathcal{S}\subseteq M$. 
 Assume that the new user's ratings also follow the linear model \eqref{withbias} with extended profile $x=(x_0,\bx)\in \{-1,+1\}\times \reals^d$. Then, the analyst can (a) infer the user's extended profile $x$, and (b) predict her ratings for items in $[M]\setminus \myset{S}$ using the extended item profiles $\{v_j\}_{j\in \myset{S}}$ as follows.
First, the analyst can infer $x$ using through the LSE:
\begin{align}\label{jointmle}
\min_{x_0\in \{-1,+1\},\bx\in \reals^d}\textstyle \sum_{j\in \myset{S}}(r_j-\langle \bx,\bv_j\rangle - x_{0}v_{j0})^2.
\end{align}
The minimization \eqref{jointmle} can be computed in  time linear in~$|\myset{S}|$, by solving two linear regressions (one for each $x_0\in\{-1,+1\}$) and picking the solution $(x_0,\bx)$ that yields the smallest error \eqref{jointmle}.  Having obtained an estimate of the  extended profile $x$,
the analyst can predict ratings as $\hat{r}_j = \langle x,v_j\rangle$, for $j\notin \mathcal{S}$.

Beyond this LSE approach, the analyst can use a different classification algorithm to first infer the private feature $x_0$, such as logistic regression or support vector machines (SVMs). We refer the reader to, \emph{e.g.}, \cite{blurme:2012}, for the description of several such algorithms and their application over real rating data. Having an estimate of $x_0$, the analyst can proceed to solve \eqref{jointmle} w.r.t.~$\bx$ alone, which involves a single linear regression.

In both of the above approaches (joint LSE, or separate inference of $x_0$) the analyst infers the private feature $x_0$. Indeed, the LSE method \eqref{jointmle} is known to predict information such as gender or age with an accuracy between 65--83\% over real datasets \cite{agenda}; separate inference of the private information (through, e.g.,  logistic regression or SVMs) leads to  84-86\% accuracy~\cite{blurme:2012}. As such, by revealing her ratings, the user also reveals  $x_0$, albeit indirectly and unintentionally.

\makeatletter{}

{
\begin{figure}[!t]
\centering
\vspace{-6pt}
\includegraphics[width=0.6\columnwidth]{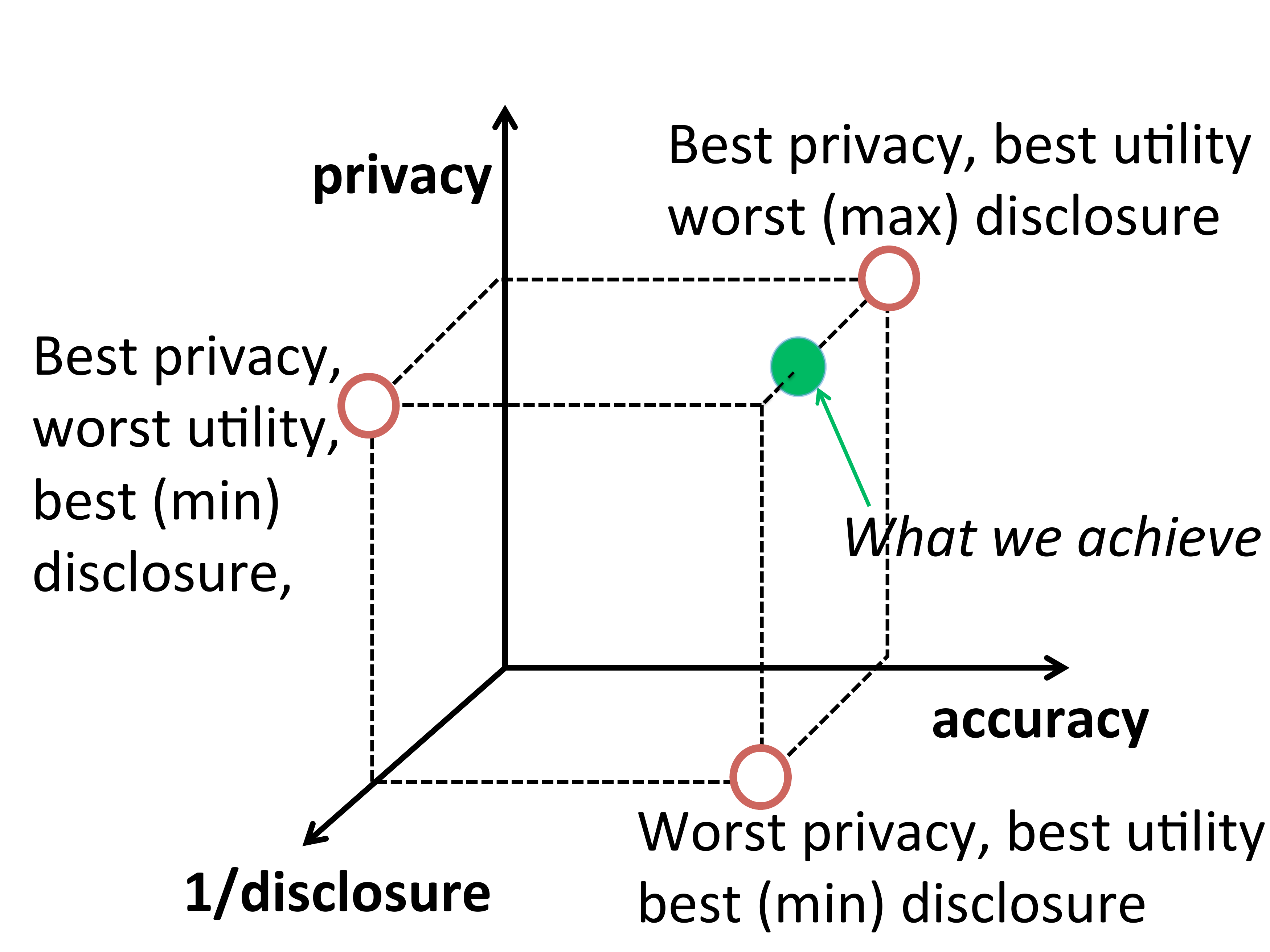}
\caption{\small
The red circles represent the three extreme protocols (Sec.~\ref{sec:overview}) that fail to meet all three properties simultaneously. Our solution (see Sec.~\ref{sec:MidpointProtocol}) lies near the upper front right edge of the cube, as it has perfect privacy and accuracy. We will prove (Thm.~\ref{privacytheorem} in Sec.~5) that the region between our solution and the optimal corner (`zero' disclosure, perfect privacy, maximal accuracy)  is unattainable. }
\label{fig:3wayTradeoff}
\end{figure}
}

\section{Modeling  Privacy Tradeoffs}\label{sec:model}

Section~\ref{sec:newuser} illustrates that serving a privacy-conscious user  is not straightforward: there is a tension between the user's privacy and the utility she receives. Accurate profiling allows correct rating prediction and enables  relevant recommendations, at the cost of the inadvertent revelation of the user's private feature. It is thus natural to ask whether the user can benefit from accurate prediction \emph{while preventing the inference of this feature}.
We will provide both rigorous and empirical evidence that -- perhaps surprisingly -- this is possible. Specific features can be obfuscated without harming personalization. One of our main contributions  is to identify that, beyond this privacy-utility tradeoff, there is in fact a third aspect to this problem: namely, how much information  the analyst discloses to the user. In what follows, we present  a framework that addresses these issues. 
\subsection{Problem Formulation} \label{sec:overview}

Motivated by Section~\ref{sec:newuser}, we consider a setting comprising the two entities we have encountered so far, an \emph{analyst} and a \emph{privacy-conscious user}. The analyst has access to a dataset of ratings collected from users that are not privacy-conscious, and have additionally revealed to the analyst a binary feature. By performing matrix factorization over this dataset, the analyst has extracted  extended item profiles $v_j=(v_{j0},\bv_j)\in \reals^{d+1}$, $j\in [M]$, for a set of $M$ items.

The analyst solicits the ratings of the privacy-conscious user for a subset of items $\myset{S}\in[M]$. We again assume that the user is parametrized by an extended profile $x=(x_0,\bx)\in \{-1,+1\}\times \reals^d$, and that her ratings follow \eqref{withbias}.
The analyst's goal is to profile the user and identify items that the user might rate highly in $[M]\setminus\myset{S}$. The user is willing to aid the analyst in correctly profiling her; however, she is privacy-conscious w.r.t.~her private feature $x_0$, and wishes to \emph{prevent its inference}.  We thus wish to design a \emph{protocol}  for exchanging information between the analyst and the user that has three salient properties; we state these here informally, postponing precise definitions until Section~\ref{properties}:

\begin{packedenumerate}
\item[(a)] At the conclusion of the protocol, the analyst  estimates $\bx$, the non-private component of $x$, \emph{as accurately as possible}.
\item[(b)] The analyst \emph{learns nothing} about the  private feature $x_0$.
\item[(c)] The user learns \emph{as little as possible} about the extended profile $v_j$ of each item $j$.
\end{packedenumerate}

To highlight the interplay between these three properties, we discuss here three ``non-solutions'', i.e., three protocols that fail to satisfy all three properties. First,  observe that the ``empty'' protocol (no information exchange) clearly satisfies (b) and (c), but not (a): the analyst does not learn $\bx$. Second, the protocol in which the user discloses her ratings to the analyst ``in the clear'', as in Section~\ref{sec:newuser}, satisfies (a) and (c) but not (b): it allows the analyst to estimate \emph{both} $\bx$ \emph{and} $x_0$ through, e.g., the LSE \eqref{jointmle}.

Finally, consider the following protocol. The analyst discloses all item profiles $v_j$, $j\in \myset{S}$,  to the user. Subsequently, the user estimates $\bx$ locally, by solving the linear regression \eqref{jointmle} over her ratings in $\myset{S}$, with her private feature $x_0$ fixed. The user concludes the protocol by sending the obtained estimate of $\bx$ to the analyst. Observe that this protocol satisfies (a) and (b), but not (c). In particular, the user learns the extended profiles of all items in their entirety.

These protocols illustrate that it is simple to satisfy any two of the above properties, but not all three. {Each of the three ``non-solutions'' above are in fact extrema among protocols constrained by (a)-(c): each satisfies two properties in the best way possible, while completely failing on the third. In the conceptual schematic of Figure~\ref{fig:3wayTradeoff} we illustrate where these three extreme protocols lie.
} 

There is  a clear motivation, from a practical perspective, to seek protocols satisfying all three properties.
Property (a) ensures that, at the conclusion of the protocol, the analyst learns the non-private component of the user's profile, and can use it to suggest new items--benefiting thusly the user, and motivating the existence of this service.   Property (b) ensures that a privacy-conscious user receives this benefit \emph{without revealing her private feature}, thereby incentivizing her participation. Finally, property (c) limits the extent at which the item profiles $\{v_j\}_{j\in \myset{S}}$ are made publicly available. Indeed, the item profiles and the dataset from which they were constructed are proprietary information:  disclosing them to any privacy-conscious user, as described by the last non-solution, would allow any user to offer the same service. More generally, it is to the analyst's interest to enable properties (a) and (b), thereby attracting privacy-conscious users, while limiting the disclosure of any proprietary information and its exposure to competition.

It is natural to ask what is the precise  statement of ``as accurately as possible'', ``learns nothing'', and ``as little as possible'' in the above description of (a)-(c).  We provide such formal definitions below.

\subsection{A Learning Protocol}\label{userfeedback}\label{learningprotocol}

\begin{figure}[!t]
\hspace*{\stretch{1}}\includegraphics[width=0.6\columnwidth]{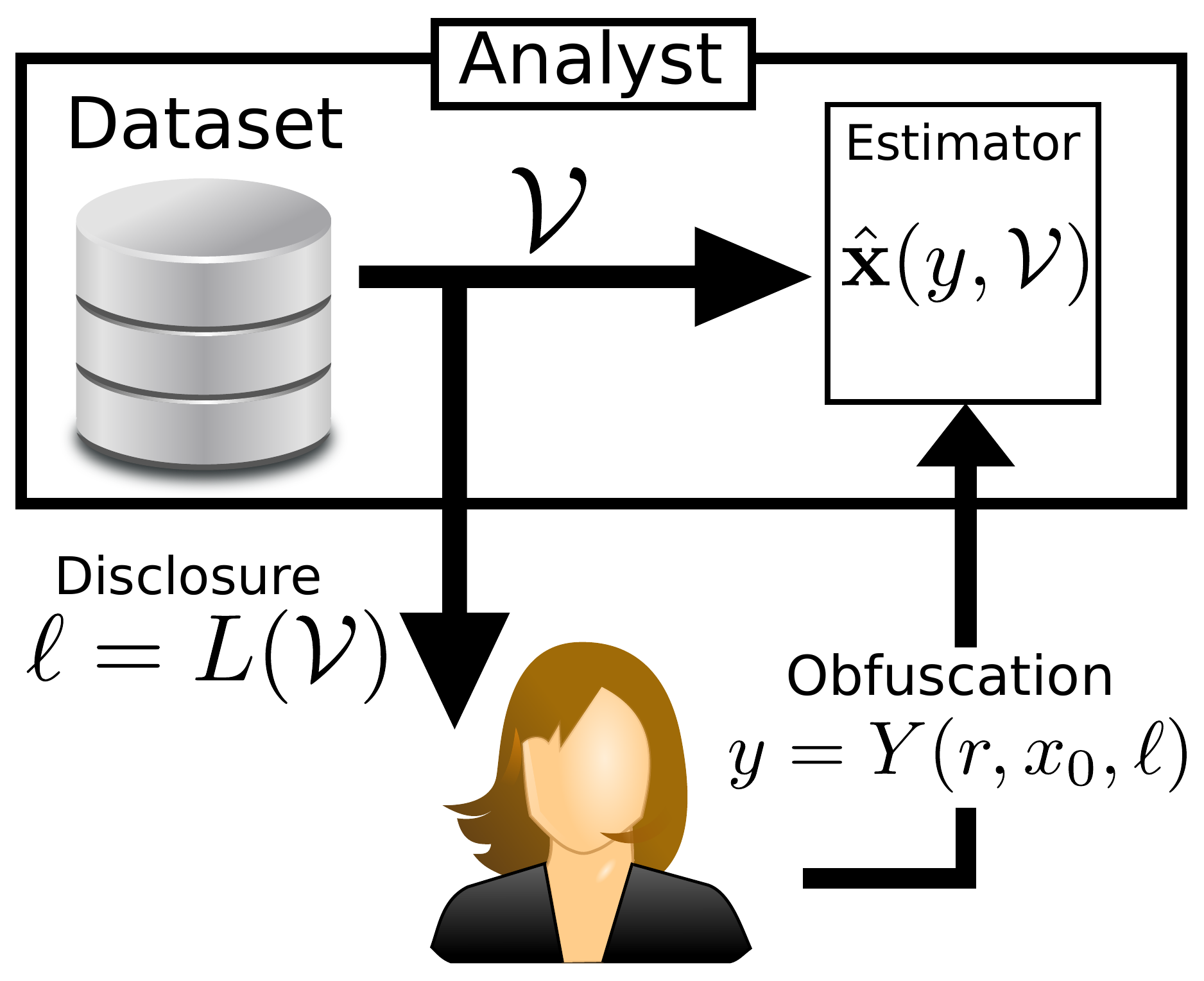}\hspace*{\stretch{1}}
\caption{\small A learning protocol $\cR = (\leak,\obfs,\bhx)$ between an analyst and a privacy-conscious user. The analyst has access to a dataset, from which it extracts the extended profiles $\cV$ through MF. It discloses to the user the information $\ell = \leak(\cV)$. Using this information, her vector of ratings $r$, and her private feature $x_0$, the user computes the obfuscated output $y = \obfs(r,x_0,\ell)$ and reveals it to the analyst. The latter uses this obfuscated feedback as well as the profiles $\cV$ to estimate $\bx$, using the estimator $\bhx(y,\cV)$.}\label{fig:setting}
\end{figure}

To formalize the notions introduced in properties (a)-(c) of Section~\ref{sec:overview}, we describe in this section the interactions between the privacy-conscious user and the analyst in a more precise fashion. Recall that the user is parametrized by an extended profile $x=(x_0,\bx)\in \{-1,+1\}\times \reals^d$, and that her ratings follow \eqref{withbias}; namely,
\begin{align}\label{eq:LinearModel}
r_{j} = \langle \bx,\bv_j\rangle  + x_{0}v_{j0} +\varepsilon_{j} = \langle x,v_j \rangle +\varepsilon_{j}, \quad j\in [M]
\end{align}
where $v_j\in\reals^{d+1}$, is the extended profile of item $j$, extracted through MF, and $\varepsilon_j$ are i.i.d.~zero mean random variables of  variance $\sigma^2<\infty$.  We note that, unless explicitly stated, we \emph{do not} assume that the noise variables $\varepsilon_j$ are Gaussian; our results will hold with greater generality.

 We assume  that the set of items $\myset{S}\subseteq [M]$, for which ratings are solicited, is an arbitrary set chosen by the analyst\footnote{We discuss how the analyst can select the items  in $\myset{S}$  in Section~\ref{sec:itemselection}.}.
 We restrict our attention to items with extended profiles $v_j$ such that $\bv_j\neq \mathbf{0}$. Indeed, given the analyst's purpose of estimating $\bf{x}$,  the rating of an item for which $\bv=\mathbf{0}$ is clearly uninformative in light of \eqref{eq:LinearModel}. We denote by
$$ \reals^{d+1}_{-\mathbf{0}} \equiv \{(v_0,\bv)\in \reals^{d+1}: \bv\neq \mathbf{0}\}$$ the set of all such vectors,
and by
 $\cV\equiv\{v_j,\; j\in \mathcal{S}\}\subseteq\reals^{d+1}_{-\mathbf{0}}$ the extended profiles of items in $\myset{S}$. Recall that the user does not  a priori know  $\cV$.
In addition, the user knows her private variable $x_0$  and either knows or can easily generate her rating $r_j$ to each item $j\in \myset{S}$. Nevertheless, the user  \emph{does not a-priori know} the remaining profile $\bx\in \reals^d$. This is consistent with MF, as the
 ``features'' corresponding to each coordinate of $\bv_j$ are  ``latent''.

The assumption that the user either knows or can easily generate her ratings in $\myset{S}$ is natural  when the user can immediately form an opinion (this is the case for items such as blog posts, ads, news articles, tweets, pictures, short videos, etc.); or, when the ``rating'' is automatically generated from user engagement (e.g., it is the time a user spends at a website, or the reading collected by a skin sensor); or, when the user is obligated to generate a response (because, e.g., she is paid to do so). We discuss the case where the user can readily produce ratings for only a subset of $\myset{S}$ in Section~\ref{sec:partial}.

Using the above notation, we define a \emph{privacy-preserving learning protocol} as a protocol consisting of the following three components, as illustrated in Figure~\ref{fig:setting}:
\begin{packeddescription}
\item[Disclosure.] The disclosure determines the amount of information that the analyst discloses to the user regarding the profiles in $\cV$. Formally, it is a mapping $$\leak:{\reals^{d+1}_{-\bf{0}}}\to \cL,$$ with $\cL$ a generic set\footnote{For technical reasons $\cL$, and $\cY$ below, are in fact measurable spaces, which include of course $\reals^k$, for some $k\in\naturals$.}.
This is implemented as a program and executed  by the analyst, who discloses to the user the information $\ell_j \equiv L(v_j)\in\cL$ for each item $j\in \myset{S}$. We denote by $L(\cV)$ the vector $\ell\in \cL^{|\myset{S}|}$ with coordinates $\ell_j$, $i\in \myset{S}$.
We note that, in practice,  $\leak(\cV)$ can be made public, as it is needed by all potential privacy-conscious users that wish to interact with the analyst.
\item[Obfuscation Scheme.] The obfuscation scheme describes how user ratings are modified (obfuscated) before being revealed to the analyst.  Formally, this is a mapping $$\obfs:\reals^{|\myset{S}|}\times\{-1,+1\}\times\cL^{|\myset{S}|}\to \cY,$$ for $\cY$ again a generic set.
The mapping is implemented as a program and executed by the user. In particular, the user enters her ratings
$r=(r_1,\ldots,r_{|\myset{S}|})\in \reals^{|\myset{S}|}$,  her private variable $x_0$ \emph{as well as} the  disclosure $\ell=\leak(\cV)\in \cL^{|\myset{S}|}$.
The program combines these quantities computing the obfuscated value
$y= \obfs(r,x_0,\ell)\in \cY$, which the user subsequently reveals to the analyst.
\item[Estimator.] Finally, using the obfuscated output by the user, and the item profiles, the analyst constructs an estimator of the user's profile $\bx\in \reals^d$. Formally: $$\bhx:\cY \times  (\reals^{(d+1)}_{-\bf{0}})^{|\myset{S}|} \to\reals^d.$$ That is, given the item feature vectors
$\cV \subset \reals^{d+1}_{-\bf{0}}$ and the corresponding obfuscated user feedback $y\in \cY$,
it yields an estimate $\bhx(y,\cV)$ of the user's non-private feature vector $\bx$.
The estimator is a program executed by the analyst.
\end{packeddescription}
We  refer to a triplet $\cR = (\leak,\obfs,\bhx)$ as a \emph{learning protocol}.
Note that the functional forms of all three of these components are known to both parties: \emph{e.g.}, the analyst knows the obfuscation scheme $\obfs$.  Both parties are \emph{honest but curious}: they follow the protocol, but if at any step they can extract more information than what is intentionally revealed, they do so.
All three mappings in protocol $\cR$  can be randomized. In the following, we denote by $\Pm_{x,\cV}$, $\Em_{x,\cV}$ the probability and expectation with respect to
the noise in \eqref{eq:LinearModel} as well as  protocol randomization, given $x$, $\cV$.

\subsection{Privacy, Accuracy, and Disclosure  Extent }\label{properties}
Having formally specified a learning protocol $\cR = (\leak,\obfs,\bhx)$,  we now define the three quality metrics we wish to attain, corresponding to the properties (a)-(c) of Section~\ref{sec:overview}.

\medskip\noindent\textbf{Privacy.}
We begin with our formalization of privacy:
\begin{definition}\label{def:privacy}
  We say that $\cR=(\leak,\obfs,\bhx)$   is \emph{privacy preserving} if the obfuscated output $Y$ is independent of $x_0$. Formally, for all $\bx\in\reals^d$,  $\cV\subseteq\reals^{(d+1)}_{-\bf{0}}$, and  $A\subseteq \cY$,\begin{align}\label{eq:privacy}
{\Pm_{\!(-1,\bx),\cV}\big(\obfs(r,\!-1,\!\ell)\!\in\! A\big)}\!=\!{\Pm_{\!(+1,\bx),\cV}\big(\obfs(r,\!+1,\!\ell)\!\in\! A\big)},
\end{align}
where $\ell=\leak(\cV)$ is the information disclosed from $\cV$, and $r\in \reals^{|\myset{S}|}$ is the vector of  user ratings.
\end{definition}
Intuitively, a learning protocol is privacy-preserving if its obfuscation scheme reveals nothing about the user's private variable: the distribution of the output $Y$ does not depend statistically on $x_0$. Put differently, two users that have the same $\bx$, but different $x_0$, output obfuscated values that are \emph{computationally indistinguishable} \cite{goldwasser}.

Computational indistinguishability is a strong privacy property, as it implies a user's private variable is protected against \emph{any} inference algorithm (and not just, e.g., the LSE  \eqref{jointmle}): in particular, no inference algorithm can estimate $x_0$ with probability better than 50\% with access to $y$ alone.

\medskip\noindent\textbf{Accuracy.} Our second definition determines a partial ordering among learning protocols w.r.t.~their accuracy, as captured by the $\ell_2$ loss of the estimation: \begin{definition}\label{def:accuracy}
We say that a learning protocol $\cR = (\leak,\obfs,\bhx)$ is more accurate than
 $\cR' = (\leak',\obfs',\bhx')$ if, for all  $\cV\subseteq \reals^{d+1}_{-\bf{0}}$, \begin{align*}
\!\!\sup_{\substack{x_0\in\{0, 1\}\\\bx\in\reals^d}}\!\!\!\!\Em_{(x_0,\bx),\!\cV}\{\|\bhx(y,\!\cV)\!-\!\bx\|^2_2\}
\le \!\!\!\sup_{\substack{x_0\in\{0, 1\}\\\bx\in\reals^d}}\!\!\!\!\Em_{(x_0,\bx),\!\cV}\{\|\bhx'(y'\!\!,\!\cV)\!-\!\bx\|^2_2\}\, ,
\end{align*}
where $y=\obfs(r,x_0,\leak(\cV))$, $y'=\obfs'(r,x_0,\leak'(\cV))$.
Further, we say that it is strictly more accurate if the above inequality holds strictly for some $\cV\subseteq \reals^{d+1}_{-\bf{0}}$.
\end{definition}
Note that the accuracy of $\cR$ is determined by  the  $\ell_2$ loss of the estimate $\bhx$ computed in  a \emph{worst-case} scenario, among all possible extended user profiles $x=(x_0,\bx)$.

This metric is natural. As we discuss in Section~\ref{sec:itemselection}, it relates to the so-called A-optimality criterion~\cite{boyd}. It also has an additional compelling motivation. Recall that $\bhx$ is used to estimate the rating for a new item through the inner product \eqref{prediction}.  An estimator $\bhx$ minimizing the expected $\ell_2$ loss also minimizes the mean square prediction error over a new item. This further motivates this accuracy metric, given that the analyst's goal is correct rating prediction.
 
To see this, assume that the extended user profile is estimated as $\hx=(\hx_0,\bhx)$ for some $\bx_0$
(for brevity we omit
the dependence on $y,\cV$).  Recall that the analyst uses this profile to predict ratings for $v\notin \cV$ using  $\hr= \<v,\hx\>$. The quality
of such a prediction is often evaluated in terms of the mean square
error (MSE): \begin{align*}
\MSE = \Em\{(r\!-\!\hr)^2\}\! \stackrel{\eqref{eq:LinearModel}}{=}\! \sigma^2\!+\!\Em\{\<v,(x\!-\!\hx)\>^2\}\, .\end{align*}
Assuming a random item vector $v$ with diagonal covariance
$\Em(v_0^2)=c_0$, $\Em(v_0\bv) = 0$, $\Em(\bv\bv^{\sT}) = c\id$, we get
$$\MSE=\sigma^2+c_0\Em\{(x_0-\hx_0)^2\}+c\Em\big\{\|x-\hx\|_2^2\big\}\, .$$
Observe that the first term is independent of the estimation.  Under a privacy-preserving protocol, the value for $\hx_0$ that minimizes the second term is 0.5, also independent of the estimation. The last term is precisely the $\ell_2$ loss.   Hence, minimizing the mean square error of the analyst's prediction  is equivalent to minimizing the $\ell_2$ loss of the estimator $\hx$. This directly motivates our accuracy definition.

\medskip\noindent\textbf{Disclosure Extent.} Finally, we define a partial ordering among learning protocols w.r.t.~the amount of information revealed by their disclosures.
\begin{definition}\label{def:disclosure}
We say that $\cR = (L,R,\bhx)$ discloses {at least as much}
 information  as $\cR' = (\leak',\obfs',\bhx')$ if
there exists a measurable mapping
$\varphi:\cL\to\cL'$ such that $$L' = \varphi\circ L$$
\emph{i.e.}, $\leak'(v)=\varphi(\leak(v))$ for each $v\in\reals^{d+1}_{-\bf{0}}$.
We say that  $\cR$ and $\cR'$ disclose the
{same amount of} information if  $\leak = \varphi\circ \leak'$ and  $\leak' = \varphi'\circ \leak$
for some $\varphi$, $\varphi'$. Finally, we say that  $\cR$ discloses {strictly
more} information than $\cR'$ if $\leak' = \varphi\circ \leak$
for some $\varphi$ but there exists no $\varphi'$ such that $\leak = \varphi'\circ \leak'$.
\end{definition}
The above definition is again natural. Intuitively, a disclosure $L$ carries at least as much information as $L'$ if  $L'$ can be retrieved from $L$: the existence of the mapping $\varphi$  implies that the user can recover $L'$ from $L$ by applying $\varphi$ to the disclosure $L(\cV)$. Put differently, having a ``black box'' that computes $L$, the user can compute $L'$ by feeding the output of this box to $\varphi$. If this is the case, then $L$ is clearly at least as informative as $L'$.

\makeatletter{}\section{An Optimal Protocol}\label{sec:MidpointProtocol}

\begin{algorithm}[t]
\begin{small}
  \caption{\textsc{Midpoint Protocol}}\label{alg:MidpointProtocol}
    \begin{algorithmic}
  \STATE \textbf{Analyst's Parameters}
  \STATE $\myset{S}\subseteq [M]$,           $\cV=\{(v_{j0},\bv_j),\; j\in \mathcal{S}\}\subseteq\reals^{d+1}_{-\mathbf{0}}$
  \STATE
  \STATE \textbf{User's Parameters}
  \STATE $x_0 \in\{-1,+1\}$,          $r=(r_1,\ldots,r_{|\myset{S}|})\in \reals^{|\myset{S}|}$
  \STATE
   \STATE  DISCLOSURE: $\ell=\leak(\cV)$
       \STATE $\ell_j = v_{j0}$, for all $j\in \myset{S}$
          \STATE
    \STATE OBFUSCATION SCHEME: $y= \obfs(r,x_0,\ell)$
       \STATE $y_j= r_j-x_0\cdot \ell_j$, for all $j\in \myset{S}$
          \STATE
    \STATE ESTIMATOR: $\bhx=\bhx(y,\cV)$
       \STATE Apply the minimax estimator $\bhx^*$ given by \eqref{eq:minimax}.  \end{algorithmic}
\end{small}
\end{algorithm}

In this section we prove that a simple learning protocol outlined in Algorithm~\ref{alg:MidpointProtocol}, which we refer to as the \emph{midpoint protocol} (MP), has remarkable optimality properties. The three components $(\leak,\obfs,\bhx)$ of MP are as follows:
\begin{packedenumerate}
\item The analyst discloses the entry $v_0$ corresponding to the private
user feature $x_0$, \emph{i.e.},  $\leak\big((v_0,\bv)\big) \equiv v_0$ for all $(v_0,\bv)\in \reals^{d+1}_{-\bf{0}}$, and  $\cL\equiv\reals$.
\item The user shifts each rating $r_j$ by the contribution of her private feature. More specifically, the user reveals to the analyst the quantities:
 $$y_j= r_j-x_0\cdot \ell_j=r_j-x_0\cdot v_{j0}, \quad j \in \myset{S}.$$
The user's obfuscated feedback is thus $\obfs(r,x_0,\ell) \equiv y$, where vector $y$'s coordinates are the above quantities, i.e., $y=(y_1,\ldots,y_{|S|})$, and  $\myset{Y}\equiv \reals^{|\myset{S}|}$. Note that, by  \eqref{eq:LinearModel}, for every $j\in \myset{S}$ the obfuscated feedback satisfies $y_j=\<\bv_j,\bx\>+\varepsilon_j$, with $\varepsilon_j$ the i.i.d.~zero-mean noise variables.
\item Finally, the analyst applies a \emph{minimax} estimator  on the obfuscated feedback.
Let $\cX$ be the set of all measurable mappings $\bhx$ estimating $\bx$ given $y$ and $\cV$ (i.e., of the form $\bhx: 
\reals^{|\myset{S}|} \times  (\reals^{(d+1)}_{-\bf{0}})^{|\myset{S}|} \to\reals^d$). Estimator
 $\bhx^*\in \cX$ is  \emph{minimax}  if it minimizes the worst-case $\ell_2$ loss, i.e.:
\begin{align}\label{eq:minimax}
\begin{split}
\!\!\sup_{{x_0\in\{0, 1\},\bx\in\reals^d}}\!\!\!\!&\Em_{(x_0,\bx),\!\cV}\{\|\bhx^*(y,\!\cV)\!-\!\bx\|^2_2\}
=\\&\!\!\! \inf_{\bhx\in \cX}\sup_{{x_0\in\{0, 1\},\bx\in\cX}}\!\!\!\!\Em_{(x_0,\bx),\!\cV}\{\|\bhx(y,\cV)\!-\!\bx\|^2_2\}.
\end{split}
\end{align}
\end{packedenumerate}

The following theorem summarizes the midpoint protocol's remarkable properties:
\begin{theorem}\label{privacytheorem}
Under the linear model \eqref{eq:LinearModel}:
\begin{packedenumerate}
\item MP is privacy preserving.
\item  No privacy preserving protocol  is strictly more accurate than MP.
\item  Any privacy preserving protocol that does not disclose at least as much information as MP is strictly less accurate.
\end{packedenumerate}
\end{theorem}
We prove the theorem below.
Its second and third statement establish formally the optimality of the midpoint protocol.
Intuitively, the second statement implies that the midpoint protocol has \emph{maximal accuracy}. No privacy preserving protocol  achieves  better accuracy: surprisingly, this is true even among schemes that \emph{disclose strictly more information than the midpoint protocol}. As such, the second statement of the theorem imples there is no reason to disclose more than $v_{j0}$ for each item $j\in \myset{S}$.

The third statement implies that the midpoint protocol engages in a \emph{minimal disclosure}: to achieve maximal accuracy, a learning protocol \emph{must disclose at least} $v_{j0}$, $j\in \myset{S}$. In fact, our proof shows that the gap between the accuracy of MP and a protocol not disclosing biases is infinite, for certain $\cV$. We note that the disclosure in MP is intuitively appealing: an analyst need only disclose the gap between average ratings across the two types (e.g., males and females, conservatives and liberals, etc.) to enable protection of $x_0$.

In general, the minimax estimator $\bhx^*$  depends on the distribution followed by the noise variables in \eqref{eq:LinearModel}. For arbitrary distributions,  a minimax estimator that can be computed in a closed form (rather than as the limit of a sequence of estimators) may not be known. General conditions for the existence of such estimators can be found, e.g., in Strasser~\cite{strasser1982local}. 
In the case of Gaussian noise, the minimax estimator coincides with the least squares estimator (see, e.g., Lehman and Casella~\cite[Thm.~1.15, Chap.~5]{lehmann}), i.e., 
\begin{align}\label{eq:lsq}
\textstyle\!\!\!\!\!\bhx^*(y,\cV)\!=\! \argmin_{\bx\in\reals^d}\! \Big\{\!\sum_{j=1}^{|\myset{S}|}\!\!\big(y_j\!-\!\<\bv_j\!,\!\bx\>\big)^2\Big\}\, .\!
\end{align}
The minimization in \eqref{eq:lsq} is a linear regression, and   $\bhx^*$ has the following closed form:
\begin{align}\label{eq:regress}
\bhx^* (y,\cV) = \textstyle\big(\sum_{j\in \myset{S}} \bv_j\bv_j^T\big)^{-1} \cdot \big(\sum_{j\in \myset{S}} y_j\bv_j\big).
\end{align}
The accuracy of this estimator can also be computed in a closed form. Using, \eqref{eq:LinearModel}, \eqref{eq:regress}, and the definition of $y$, it can easily be shown that, for all $\bx\in \reals^d$, \begin{align}\label{lsqloss}
\Em_{(x_0,\bx),\!\cV}\{\|\bhx^*(y,\cV)\!-\!\bx\|^2_2\} \! =\! \textstyle\sigma^2\mathrm{tr}\big[\big(\sum_{j\in \myset{S}} \bv_j\bv_j^T\big)^{-1}\big],
\end{align}
where $\sigma^2$ the noise variance in \eqref{eq:LinearModel} and $\mathrm{tr}(\cdot)$ the trace.

\makeatletter{}\subsection{Proof of Theorem~{\protect{\ref{privacytheorem}}}}

\noindent\textbf{Privacy.} To see that Thm.~\ref{privacytheorem}.1 holds, observe that the user releases $y_j =r_j-v_{0j}x_0 \stackrel{\eqref{eq:LinearModel}}{=} \<\bv_j,\bx\>+\varepsilon_j,$ for each $j\in \myset{S}$. The distribution of $y_j$ thus does not depend on $x_0$, so the midpoint protocol is clearly privacy preserving.

\medskip\noindent\textbf{Maximal Accuracy.} We prove  Theorem~\ref{privacytheorem}.2 by contradiction; in particular, we show that a protocol that is strictly more accurate can be used to construct an estimator that has lower worst-case $\ell_2$ loss than the minimax estimator. 

Suppose that there exists a privacy preserving protocol $\cR'=(\leak',\obfs',\bhx')$ that is strictly more accurate than the midpoint protocol $\cR=(\leak,\obfs,\bhx)$. Let $\ell=\leak(\cV),\ell'=\leak'(\cV)$ be the disclosures under the two protocols, and $y = \obfs(r,x_0,\ell)$, $y'=\obfs'(r,x_0,\ell')$ the obfuscated values. Recall that  $$\ell_j=v_{j0},~\text{and}~ y_j = r_j-x_0v_{0j} = \<\bv_j,\bx\>+\varepsilon_j, \quad j\in \myset{S}.$$      
We will use $\leak'$, $\obfs'$ and $\bhx'$ to construct an estimator $\bhx''$ that has a lower $\ell_2$ loss than the least squares estimator $\bhx$ over $y$ and $\cV$.  First, apply $\obfs'$ to $y+\ell$, assuming that the private variable is $x_0=+1$, using the disclosed information $\ell'$. That is:
$y'' = \obfs'(y+\ell,+1,\ell').$
 Second,  apply the estimator $\bhx'$ to this newly obfuscated output $y''$, i.e.:
$\bhx'(y'',\cV) $ 
Combining these two the estimator $\bhx''$ is given by
\begin{align*}
\bhx''(y,\cV) = \bhx'\left( \obfs'\left(y+\ell,+1,\leak'(\cV)\right), \cV\right) 
\end{align*} 
Under this construction, the random variables $y''$, $y'$ are identically distributed. This is obvious if $x_0=+1$; indeed, in this case $y''=y'$. On the other hand, since  $\cR'$ is privacy preserving, by \eqref{eq:privacy}:
\begin{align}\obfs'(y+\ell,+1,\ell') \stackrel{{\rm d}}{=} \obfs'(y-\ell,-1,\ell'),\label{eq:same}\end{align}
{i.e.},  the two variables are equal in distribution.

This implies that $\bhx''(y,\cV)$ is identically distributed as $\bhx'(y',\cV)$. On the other hand, $\cR'$ is strictly more accurate than $\cR$; hence, there exists a $\cV$ such that
\begin{align*}
\sup_{x} \Em\{\|\bhx(y,\cV)-\bx\|_2^2 \} &>  \sup_{x} \Em\{\|\bhx'(y',\cV)-\bx\|_2^2 \} \\ &= \sup_{x} \Em\{\|\bhx''(y,\cV) -\bx \|_2^2\} ,
\end{align*}
a contradiction.

\medskip\noindent\textbf{Minimal Disclosure.} Consider a privacy preserving learning protocol $\cR'=(\leak',\obfs',\bhx')$ that does not disclose at least as much information as the midpoint protocol $\cR=(\leak,\obfs,\bhx)$. Consider a setup where $|\myset{S}| = d$, the dimension of the feature profiles. Assume also that $\cV$ is such that the matrix $V = [\bv_j]_{j\in \myset{S}}\in \reals^{d\times d}$ is invertible, and denote by $\ell=\leak(\cV)\in\reals^d$ the vector with coordinates $v_{j0}$, $j\in \myset{S}$. 

For any
$x_0\in\{+1,-1\}$, $s\in \reals^d$, and $\ell'\in(\cL')^d$, let $Z_{x_0}(s,\ell')\in \cY'$ be a random variable  with distribution
given by
$Z_{x_0}(s,\ell') \ed \obfs'(s+\varepsilon,x_0,\ell'),$
where $\varepsilon\in \reals^d$ a vector of \emph{i.i.d.}~coordinates sampled from the same distribution as the noise variables $\varepsilon_j$, $j\in \myset{S}$.
Put differently, $Z_{x_0}(s,\ell')$ is distributed as the output of obfuscation $\obfs'$ when $r-\varepsilon = V\bx+x_0\ell=s\in \reals^d$,  $\leak'(\cV)=\ell'$, and the gender is $x_0$.
The following then holds.
\begin{lemma}\label{lemma:LeakPlusMinus}
If $V\in \reals^{d\times d}$ is invertible then, for all $s\in\reals^d$, $\ell=\leak(\cV)$, and $\ell' = \leak'(\cV)$,
$
Z_+(s, \ell') \ed Z_-(s -2\ell,\ell').$
\end{lemma}
\begin{proof}
 By Eq.~(\ref{eq:same}), for all $\bx\in\reals^d$, $$\obfs'(V\bx+\ell+\varepsilon,+1,\ell')
\ed \obfs'(V\bx-\ell+\varepsilon,-1,\ell')\, .$$
The claim follows by taking $\bx = V^{-1}(s-\ell)$.\hspace*{\stretch{1}}
\end{proof}

As $\cR'$ does not disclose as much information as the midpoint protocol, by definition, there is no map $\varphi$ such that $\leak(v)=\varphi(\leak'(v))$ for all $v=(v_0,\bv)\in \reals^{d+1}_{-\bf{0}}$. Hence, there exist extended profiles $v,v'\in \reals^{d+1}_{-\bf{0}}$ such that $v_{0}\neq v_{0}'$ and yet $\leak'(v) = \leak'(v')$. 
As both $v=(v_0,\bv),v'=(v_0',\bv')$ belong to $\reals^{d+1}_{-\bf 0}$, the supports of $\bv,\bv'$ are non-empty.
We consider the following two cases:

\textbf{Case 1.} The supports of $\bv,\bv'$ intersect, i.e., there exists a $k\in [d]$ such that $v_k\neq 0$ and $v_k'\neq 0$. In this case, 
consider a scenario in which $\cV=\{v\}\cup\bigcup_{1\leq l\leq d,l\neq k,}\{e_l\}$, where $e_l\in\reals^{d+1}_{-\bf{0}}$ a vector whose $l$-th coordinate is 1 and all other coordinates are zero. Clearly, $|\myset{S}|=|\cV|=d$, and $V=[\bv_i]_{i\in [d]}$ is invertible. Let $\ell^*=\leak'(\cV)$.
By Lemma \ref{lemma:LeakPlusMinus},
for all $s\in\reals$,
\begin{align}Z_+(s+2v_0{\bf e}_1,\ell^*)\ed Z_-(s,\ell^*)\, ,\label{withv}\end{align}
where 
 ${\bf e}_1\in\reals^d$ is 1 at coordinate $1$ and 0 everywhere else.
Similarly, in a scenario in which $\cV'=\{v'\}\cup\bigcup_{1\leq l\leq d,l\neq k,}\{e_l\}$, $V$ is again invertible. Crucially $\leak'(\cV')=\leak(\cV)=\ell^*$, so again  by Lemma~\ref{lemma:LeakPlusMinus}:
\begin{align}Z_+(s+2v_0'{\bf e}_1,\ell^*)\ed Z_-(s,\ell^*)\, ,\label{withv'}\end{align}
for all $s\in\reals^d$. 
Equations \eqref{withv},\eqref{withv'} imply that, for all $s\in \reals^d$:
\begin{align}
Z_+(s+\xi {\bf e}_1 ,\ell^*) \ed Z_+(s,\ell^*)\, \label{periodic}
\end{align}
where $\xi\equiv 2(v_{0} - v_{0}').$ In other words, the obfuscation is \emph{periodic} with respect to the direction ${\bf e}_1$.

Observe that for any $x\in\{-1,+1\}\times \reals^d$ and any $M\in \reals_+$,  we can construct a $x'\in\{-1,+1\}\times \reals^d$ and a $K\in \naturals$ such that (a) $x$,$x'$ differ only at  coordinate $k\in\{1,2,\ldots,d\}$, (b) $\<v,x-x'\> = K\xi$, and (c) $\|\bx-\bx'\|_2\geq M$. To see this, let $K$ be a large enough integer such that $\frac{K|\xi|}{|v_{k}|}>M$. Taking, $x_k'=x_k+K\xi/v_{k}$, and  $x'_l =x_l$ for all other $l$ in $\{0,1,\ldots,d\}$ yields a $x'$ that satisfies the desired properties (a) and (b).

\sloppy
Suppose that the learning protocol $\cR$ is applied to $\cV=\{v\}\cup\bigcup_{1\leq l\leq d,\neq k}\{e_l\}$ for a user with $x_0=+1$.  Fix a large $M>0$. For each $x$ and $x'$ constructed as above, by \eqref{periodic}, the obfuscated values generated by $\obfs'$  have an identical distribution.  Hence, irrespectively of how the estimator $\bhx'$ is implemented, either $\Em_{x,\cV}\{\|\bhx'(y',\cV)-\bx\|^2_2\}$ or $
\Em_{x',\cV}\{\|\bhx'(y',\cV)-\bx'\|^2_2\}$ must be $\Omega(M^2)$
which, in turn, implies that
$\sup_{x\in\{\pm 1\}\times\reals^d}\Em_{x,\cV}\{\|\bhx'(y',\cV)-\bx\|^2_2=\infty$.

Note that, since $\bv_j, j \in \myset{S},$ are linearly independent, the matrix $\sum_{j\in \myset{S}}\bv_j\bv_j^T$ is positive definite and thus invertible. Hence, in contrast to the above setting, the loss  \eqref{lsqloss} of MP in the case of Gaussian noise is finite.\\
\fussy
\textbf{Case 2}. The supports of $\bv,\bv'$ are disjoint. In this case $v,v'$ are linearly independent and, in particular, there exist $1\leq k,k'\leq d $, $k\neq k'$, such that  $v_k\neq 0$, $v_{k'}=0$ while  $v'_k=0$, $v'_{k'}\neq 0.$ Let  $\cV=\{v\}\cup\{v'\}\bigcup_{1\leq l\leq d,l\neq k,l\neq k'}\{e_l\}$. Then, $|\cV|=d$ and the matrix $V=[\bv_j]_{j\in \myset{S}}$ is again invertible. By swapping the positions of $\bv$ and $\bv'$ in matrix $V$ we can show using a similar argument as in Case 1 that
 for all $s\in \reals^d$:
\begin{align}
Z_+(s+\xi ({\bf e}_1 - {\bf e}_2),\ell^*) \ed Z_+(s,\ell^*)\, 
\label{periodic2}
\end{align}
where $\xi\equiv 2(v_{0} - v_{0}')$ and $\ell^*=\leak(\cV)$. \emph{I.e.}, $Z_+$ is periodic in the direction ${\bf e}_1-\bf{e}_2$. Moreover, 
for any $x\in\{-1,+1\}\times \reals^d$ and any $M\in \reals_+$,  we can similarly construct a $x'\in\{-1,+1\}\times \reals^d$ and a $K\in \naturals$ such that (a) $x$,$x'$ differ only at  coordinates $k,k'\in\{1,2,\ldots,d\}$,  and (b) $\<v,x-x'\> = -\<v',x-x'\>= K\xi$, and (c) $\|\bx-\bx'\|_2\geq M$: the construction  adds $K\xi/v_{k}$ at the $k$-th coordinate and subtracts $K\xi/v_{k'}'$ from the $k'$-th coordinate, where $K>M\max(v_k,v'_{k'})/\xi$.  A similar argument as in Case 1 can be used to show again that the estimator $\bhx'$ cannot disambiguate between $x$, $x'$ over $\cV$, yielding the theorem.\hspace*{\stretch{1}}  \qed

\makeatletter{}\section{Extensions}\label{sec:Extensions}\label{sec:Discussion}

We have up until now assumed that the analyst solicits ratings for a set of items $\myset{S}\subseteq [M]$, determined by the analyst before the user reveals her feedback.   In what follows, we discuss how our analysis can be extended in the case where the user only provides ratings for a subset of these items. We also discuss how the analyst should select $\myset{S}$, how a user can repeatedly interact with the analyst, and, finally, how to deal with multiple binary  and categorical features. 

\subsection{Partial Feedback}\label{sec:partial}

\begin{algorithm}[t]
\begin{small}
  \caption{\textsc{Midpoint Protocol with Sub-Sampling}}\label{alg:SelectionProtocol}
    \begin{algorithmic}
  \STATE \textbf{Analyst's Parameters}
  \STATE $\myset{S}\subseteq [M]$, $\cV=\{(v_{j0},\bv_j),\; j\in \mathcal{S}\}\subseteq\reals^{d+1}_{-\mathbf{0}}$
  \STATE $p=\{(p_1^{x_0},\ldots,p_{|\myset{S}|}^{x_0}), x_0 \in\{-,+\} \} \subseteq ([0,1]\times [0,1])^{|\myset{S}|}$
      \STATE
  \STATE \textbf{User's Parameters}
  \STATE $x_0 \in\{-1,+1\}$, $\myset{S}_0\subseteq \myset{S}$,  $r=\{r_j,j\in \myset{S}_0\}\in \reals^{|\myset{S}_0|}$
  \STATE
   \STATE DISCLOSURE: $\ell=\leak(\cV,p)$
   \STATE $\rho_j={p_j^{-}}/{p_j^{+}}$, for all $j\in \myset{S}$
    \STATE $\ell_j = (v_{j0},\rho_j)$, for all $j\in \myset{S}$
    \STATE
    \STATE OBFUSCATION SCHEME: ${{\myset{S}_R=\myset{S}_R(\myset{S}_0,x_0,\ell)}\atop{ y= \obfs(r_{\myset{S_R}},x_0,\ell)}}$
    \STATE $S_R=\emptyset$,  $y=\emptyset$
    \FORALL{$j \in \myset{S}$}
    \IF{$j\in  \myset{S}_0$}
    \STATE $b_j \sim \mathrm{Bern}\big(\min\big(1,(\rho_j)^{x_0}\big)\big)$
    \IF {$b_j =1$}
    \STATE $\myset{S}_R = \myset{S}_R\cup\{j\}$
    \STATE $y = y\cup\{r_j - x_0v_{j0}\}$
    \ENDIF
    \ENDIF
    \ENDFOR
    \STATE
    \STATE ESTIMATOR: $\bhx=\bhx(y,(\myset{S}_R,y),\cV)$
    \STATE Solve $\bhx = \argmin_{\bx\in\reals^d}\! \Big\{\!\sum_{j\in \myset{S}_R}\!\!\big(y_j\!-\!\<\bv_j,\bx\>\big)^2\Big\}$
          \end{algorithmic}
\end{small}
\end{algorithm}

There are cases of interest  where a user may not be able to generate a rating for all items in $\myset{S}$. This is especially true when the user needs to spend a non-negligible effort to determine her preferences (examples include rating a feature-length movie, a restaurant, or a book). In these cases, it makes sense to assume that a user may readily provide ratings for only a set $\myset{S}_0\subseteq \myset{S}$ (e.g., the movies she has already watched, or the restaurants she has already dined in, etc.). 

Our analysis up until now applies when the user rates an arbitrary set $\myset{S}$ \emph{selected by the analyst}. As such, it does not readily apply to this case: the set of items $\myset{S}_0$ a user rates may depend on the private feature $x_0$ (e.g., some movies may be more likely to be viewed by men or liberals). In this case, $x_0$ would be be inferable not only from the ratings she gives, but also \emph{from which items she has rated}. 

In this section, we describe  how to modify the midpoint protocol to deal with this issue. Intuitively, to ensure her privacy, rather than reporting obfuscated ratings for \emph{all} items she rated (i.e., set $\myset{S}_0$), the user can reveal ratings \emph{only for a subset $\myset{S}_R$ of $\myset{S}_0$}. This \emph{sub-sampling} of $\myset{S}_0$ can be done so that  $\myset{S}_R$ has a distribution that is \emph{independent of $x_0$}, even though $\myset{S}_0$ does not. Moreover, to ensure a high estimation accuracy, the user ought to ensure $\myset{S}_R$ is as large as possible, subject to the constraint $\myset{S}_R\subseteq \myset{S}_0$. 

\smallskip\noindent\textbf{Model.} Before we present the modified protocol, we describe our assumption on how $\myset{S}_0$ is generated. For each $j\in [M]$, denote by $p^{+}_j$, $p^-_j$ the probabilities that a user with private feature $+1$ or $-1$, respectively, has rated item $j \in \myset{M}$. Observe that, just like the extended profiles $v_j$, this information can be extracted from a dataset comprising ratings by non privacy-conscious users. Let  $p = [(p^+_j,p^-_j)]_{j\in \myset{S}}\in  ([0,1]\times [0,1])^{|\myset{S}|} $ be the vector of pairs of probabilities.

We assume that the privacy-conscious user the has rated items in the set $S_0\subseteq \myset{S}$, whose distribution is given by the product form\footnote{Here, we slightly abuse notation, e.g., denoting with $p^{x_0}_j$ the parameter $p^{+}_j$ when $x_0=+1$.}:
\begin{align} \Pm_{x,\cV, p}(\myset{S}_0 = \myset{A}) = \prod_{j\in \myset{S}} p_j^{x_0}\!\!\!\!\! \prod_{j\in \myset{S}\setminus \myset{A}}\!\!\! (1- p_j^{x_0}), ~\text{for all}~\myset{A}\subseteq \myset{S}.  \label{product}\end{align}
Put differently, items $j\in \myset{S}$ are rated independently, each with a probability $p_j^{x_0}$.
Conditioned on $\myset{S}_0$, we assume that the user's ratings $r_j$, $j\in \myset{S}_0$, follow the linear model \eqref{eq:LinearModel} with Gaussian noise.  
 Note that the distribution of $\mathcal{S}_0$ depends on $x_0$: e.g., items $j$ for which $p_j^+ >p_j^-$ are more likely to be rated when $x_0=+1$. 

\smallskip\noindent\textbf{Midpoint Protocol with Sub-Sampling.}
We now present a modification of the midpoint protocol, which we refer to as the \emph{midpoint protocol with sub-sampling} (MPSS).
MPSS is  summarized in Algorithm~\ref{alg:SelectionProtocol}. 
First, along with the disclosure of the biases $v_{j0}, j \in \myset{S}$, the analyst also discloses the ratios $\rho_j \equiv p_j^-/p_j^+$, for $j\in \myset{S}$. Having access to this information, the 	user sub-samples items from $\myset{S}_0$; each  item $j\in S_0$ is included in the revealed set $\myset{S}_R$ independently with probability:
\begin{align}\Pm_{x,\cV,p}(i\in \myset{S}_R\mid j\in S_0) = \min\left(1,(\rho_j)^{x_0}\right)\, .\label{select}\end{align}
Having constructed $\myset{S}_R$, the user reveals ratings for $j\in S_R$ after subtracting $x_0v_{j0}$, as in  MP. Finally, the analyst estimates $\bhx$ through a least squares estimation over the obfuscated feedback, as in MP in the case of Gaussian noise. 
To gain some intuition behind the selection of the set $\myset{S}_R$, observe by \eqref{product} and \eqref{select} that, for any $j\in \myset{S}$,
\begin{align}\Pm_{x,\cV,p}(j \!\in\! \myset{S}_R)\! =\!p_j^{x_{0}} \min\big(1, (p_j^-\!/\!p_j^+)^{x_0}\big)\!  =\! \min (p_j^+\!,p_j^-).\label{mpssprob}\end{align}
This immediately implies that MPSS is privacy preserving: both the distribution of $\myset{S}_R$ and of the obfuscated ratings $y$ do not depend on $x_0$. In fact, it is easy to see that since $\myset{S}_R\subseteq\myset{S}_0$ \emph{any privacy preserving protocol must satisfy $\Pm_{x,\cV,p}(j \in \myset{S}_R)\leq \min (p_j^+\!,p_j^-)$}: indeed, if for example $p_j^+<p_j^-$, then a user rating $j$ with probability higher than $p_j^+$ must have $x_0=-1$ (see \techreport{our technical report \cite{techrep}}{Lemma~\ref{lem:atmost} in the appendix} for a formal proof of this statement). As such, MPSS reveals ratings for a set $\myset{S}_R$ of \emph{maximal size}, in expectation.

This intuition can be used to establish the optimality of MPSS among a wide class of learning protocols, under \eqref{eq:LinearModel} (with Gaussian noise) and \eqref{product}. We can again show that it attains optimal accuracy. Moreover, it also involves a minimal disclosure: a protocol that does not reveal the ratios $\rho_j$, $j\in \myset{S}$, necessarily rates strictly fewer items than MPSS, in expectation. We provide a formal proof of these statements, as well as a definition of the class of protocols we consider, in \techreport{our technical report \cite{techrep}}{ Appendix~\ref{app:proofofsub}}.

\subsection{Item Set Selection}\label{sec:itemselection}
Theorem~\ref{privacytheorem} implies that the analyst \emph{cannot} improve the prediction of the private variable $x_0$ through its choice of $\myset{S}$, under the midpoint protocol. In fact, the same is true under any privacy-preserving learning protocol: irrespectively of the analyst's choice for $\myset{S}$, the  obfuscated feedback $y$ will be statistically independent of $x_0$. 

The analyst can however strategically select $\myset{S}$ to effect the accuracy of the estimate of the non-private profile $\bx$. Indeed,  the analyst should attempt to select a set $\myset{S}$
that maximizes the  accuracy of the estimator $\bhx$. In settings where least squares estimator is minimax (e.g., when  noise is Gaussian), there are well-known techniques for addressing this problem. Eq.~\eqref{lsqloss} implies that  it is natural to select $\myset{S}$ by solving  
\begin{align}
\begin{split}
\text{Maximize:} &  \quad F(\myset{S}) = -\mathrm{tr}\big[\big(\textstyle\sum_{j\in \myset{S}} \bv_j\bv_j^T\big)^{-1}\big] \\
\text{subject to:}& \quad |\myset{S}| \leq B,  \myset{S} \subseteq [M],
\end{split}\label{eq:optselect}
\end{align}
where $B$ is the number of items for which the analyst solicits feedback. The optimization problem \eqref{eq:optselect} is NP-hard, and has been extensively studied in the context of experimental design (see, e.g., Section 7.5 of \cite{boyd}). The objective function $F$ is commonly referred to as the \emph{A-optimality criterion}. Convex relaxations of \eqref{eq:optselect} exist when $\myset{S}$ is a multiset, i.e., when items with the same profile can be presented to the user multiple times, each generating an i.i.d.~response \cite{boyd}. When such repetition is not possible, constant approximation algorithms can be constructed based on the fact that $F$ is increasing and submodular (see, e.g., \cite{Friedland20133872}). In particular, given any set $\myset{S}^*\subset [M]$ of  items whose profiles are linearly independent, there exists a polynomial time algorithm for maximizing  $F(\myset{S}\cup \myset{S}^*)-F(\myset{S}^*)$ subject to $|\myset{S}|\leq B$ within a  $1-\frac{1}{e}$ approximation factor~\cite{nemhauser}.

\subsection{Repeated Interactions}\label{sec:repeated}
Our analysis (and, in particular, the optimality of our protocols) persists even when the user repeatedly interacts with the analyst. In particular, the user may return to the service multiple times, each time asked to rate a different set of items $\myset{S}^{(k)}\setminus [M]$, $k\geq 1$. The selection of the set $\myset{S}^{(k)}$ could be \emph{adaptive}, i.e., depend on the obfuscated feedback the user has revealed up until the $k-1$-th time. For each $k$, the analyst again would apply MP (or MPSS), again disclosing the same information for each $\myset{S}^{(k)}$, the only difference being that the estimator $\bhx$ would be applied to \emph{all} revealed obfuscated ratings $y^{(1)},...,y^{(k)}$. This repeated interaction is still perfectly private: the joint distribution of the obfuscated outputs $y^{(k)}$ does not depend on  $x_0$. Moreover,  each estimation remains maximally accurate at each interaction, while each disclosure is again minimal.

\makeatletter{}\subsection{Categorical Features}\label{appendix:categorical}\sloppy
 We discuss below how to express categorical features as multiple binary features through binarization, and illustrate how to incorporate both cases in our analysis. The standard approach to incorporating a categorical feature $x_0\in \{1,2,\ldots,K\}$ in matrix factorization is through  category-specific biases (see, e.g., \cite{Koren:2009}), i.e., \eqref{eq:LinearModel} is replaced by
\begin{align}\label{eq:cat}
r = \langle \bx,\bv_j\rangle  + b_j^{x_0} +\varepsilon_{j}, \quad j\in [M] 
\end{align}
where $b_j^k\in\reals$, $k\in [K]$ are \emph{category-dependent} biases.  Consider a representation of $x_0\in [K]$ as a binary vector $\bx_0\in \{-1,+1\}^K$ whose $x_0$-th coordinate is $+1$, and all other coordinates are $-1$. I.e., the coordinate $\bx_{0k}$ at $k\in [K]$ is given   $+1$ if $k=x0$ and $-1$ o.w. 
For $k\in [K]$, let $\bb_{jk} \equiv b_j^k/2$ and define $\mu_j \equiv \sum_{k\in [K]}b_j^k/2$. 
Then, observe that \eqref{eq:cat} is equivalent to
\begin{align}
r &= \<\bx,\bv_j\> +\sum_{k\in [K]} \bx_{0k} \bb_{jk} +\mu_j+\varepsilon_j,\nonumber\\
& = \<\bx',\bv_j'\>+ \sum_{k\in [K]} \bx_{0k} \bb_{jk} +\varepsilon_j,\quad j \in [M],\label{eq:binarized}    
\end{align}
where $\bx'=(\bx,1)\in\reals^{d+1}$ and $\bv_j' =(\bv_j,\mu_j)\in \reals^{d+1}$.

\fussy
Hence, a categorical feature can be incorporated in our analysis as follows. First, given a dataset of ratings by non-privacy conscious users that reveal their categorical feature $x_0\in [K]$, the analyst first ``binarizes'' this feature, constructing a vector $\bx_0\in \{-1,1\}^K$ for each user. It then performs matrix factorization on the ratings using \eqref{eq:binarized}, learning vectors $\bv_j'\in \reals^{d+1}$, and biases $\bb_j=(\bb_{jk})_{k\in [K]}\in \reals^{K}$. A privacy-conscious user subsequently interacts with the analyst using the standard scheme as follows. The analyst discloses the biases $\bb_j$ for each $j\in \myset{S}$, and the user reveals $y_j = r_j-\sum_{k\in [K]} \bx_{0k} \bb_{jk}$, $j\in \myset{S}$, where $\bx_{0k}$ is her binarized categorical feature. Finally, the analyst infers $\bx'$ through linear regression over the pairs $(y_j,\bv'_j)$, $j\in \myset{S}$.

\makeatletter{}
\begin{figure*}[!t]
\parbox{2.2in}{\ignore{\setlength{\unitlength}{0.32\textwidth}
	\begin{picture}(0.7,0.7)
			\put(0,0){\includegraphics[width=0.25\textwidth]{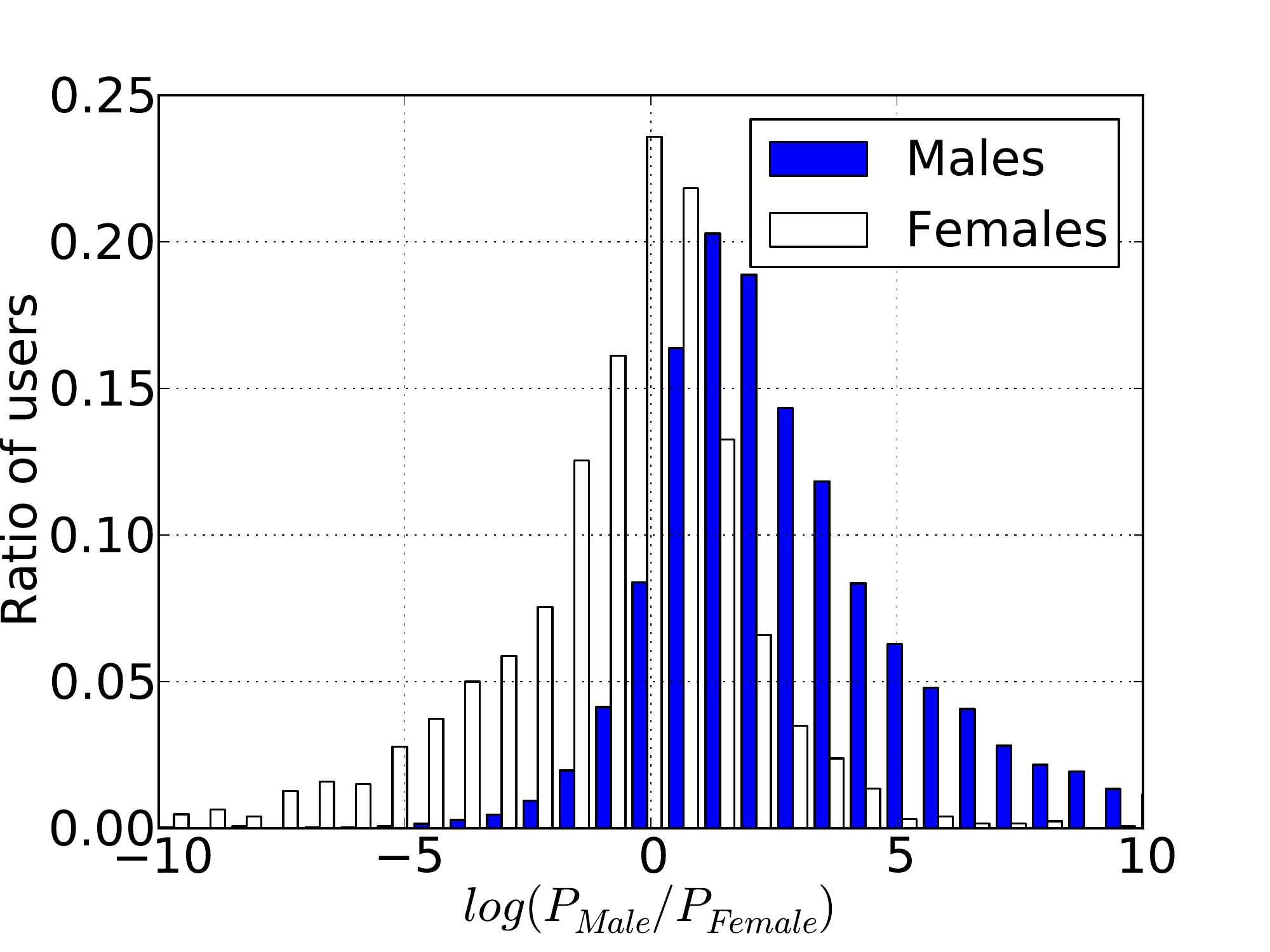}}
			\put(0.35,0.55){\tiny (a)}
		\end{picture}	
          \begin{picture}(0.7,0.3)
			\put(0,0){\includegraphics[width=0.25\textwidth]{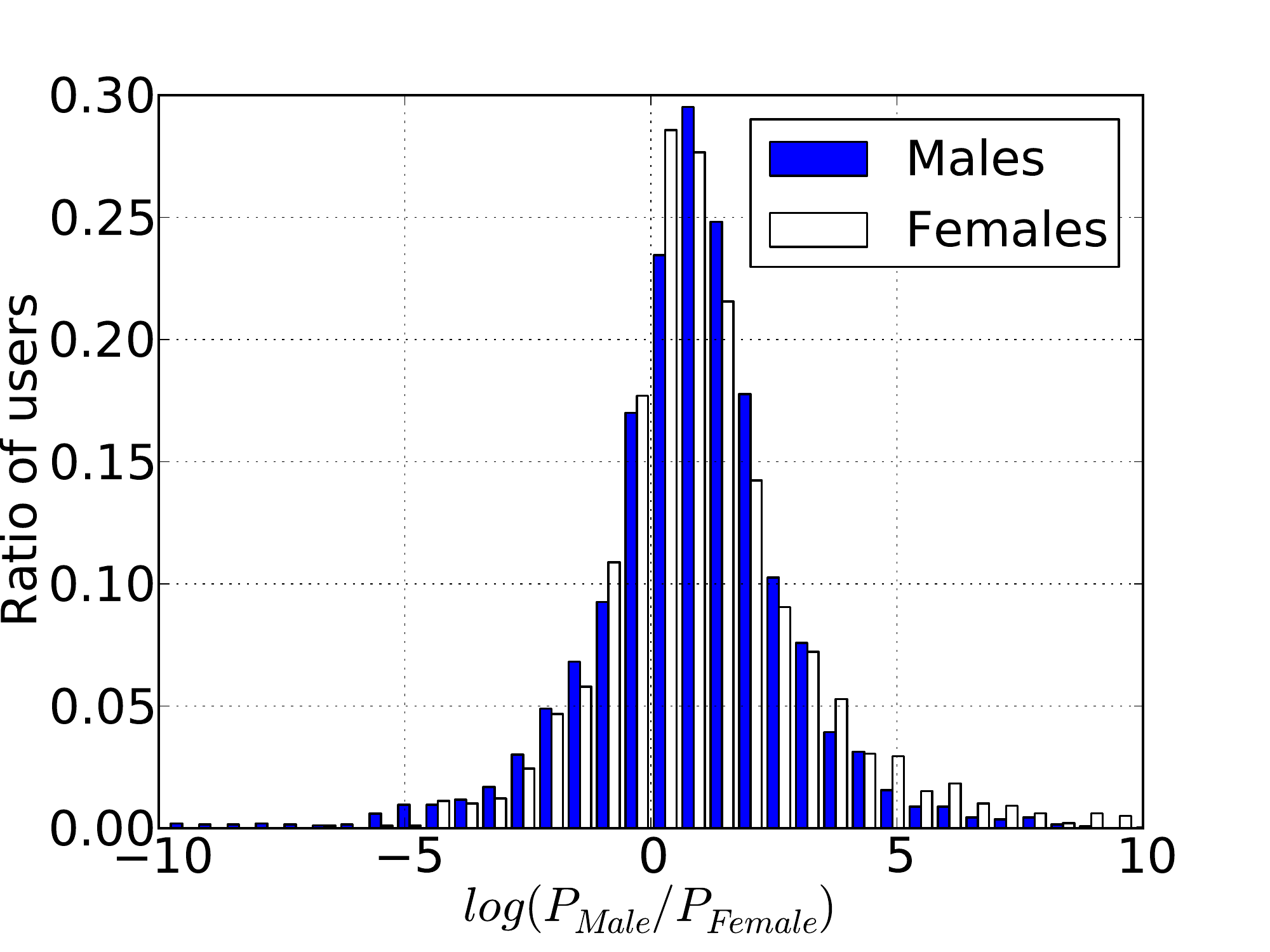}}
			\put(0.35,0.55){\tiny(b)}
		\end{picture}
   \caption{\small Distribution of inference probabilities for males and females using \ml dataset and logistic regression (a) before obfuscation and (b) after MPSS obfuscation}
    \label{fig:hists}
}
\scriptsize
\begin{center}
\begin{tabular}{|@{}l@{}|@{}l@{}|@{}c@{}|@{}c@{}|@{}c@{}|}
\hline
Dataset& Private feature & Users & Items & Ratings \\
\hline
 & All & 365 & 50 & 18K \\
PTV & Gender (F:M) & 2.7:1 & - & 2.7:1 \\
& Politics (R:D) & 1:1.4 & - & 1:1.4 \\
\hline
  & All & 6K & 3K & 1M \\
\ml & Gender (F:M) & 1:2.5 & - & 1:3 \\
& Age (Y:A) & 1:1.3 & -  &1:1.6 \\
\hline
 & All & 26K & 9921 & 5.6M \\
\fl & Gender (F:M) & 1.7:1 & - & 1.5:1 \\
\hline
\end{tabular}
\caption{\small Statistics of the datasets used for evaluation. The ratios represent the class skew
Females:Males (F:M) for gender, Young:Adult(Y:A) for age and Republican:Democrats (R:D)
for political affiliation.}
\label{table:stats}
\end{center}
\vspace{-5mm}
}\qquad
\begin{minipage}{0.666666666666666666\textwidth}
\subfloat[PTV - Gender]{
\includegraphics[width=0.5\textwidth]{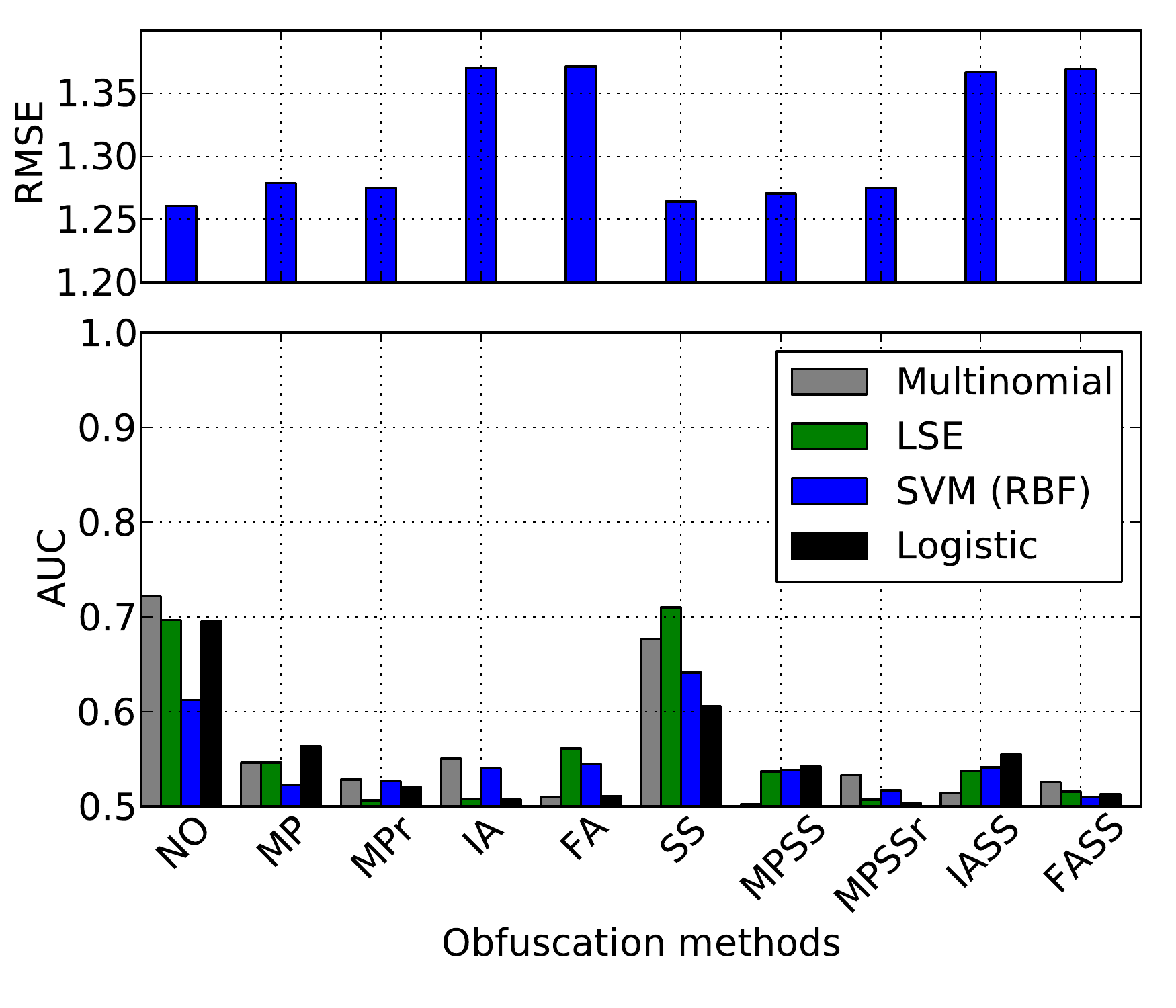}
\label{fig:ptv-gender}
}\subfloat[PTV - Politics]{
\includegraphics[width=0.5\textwidth]{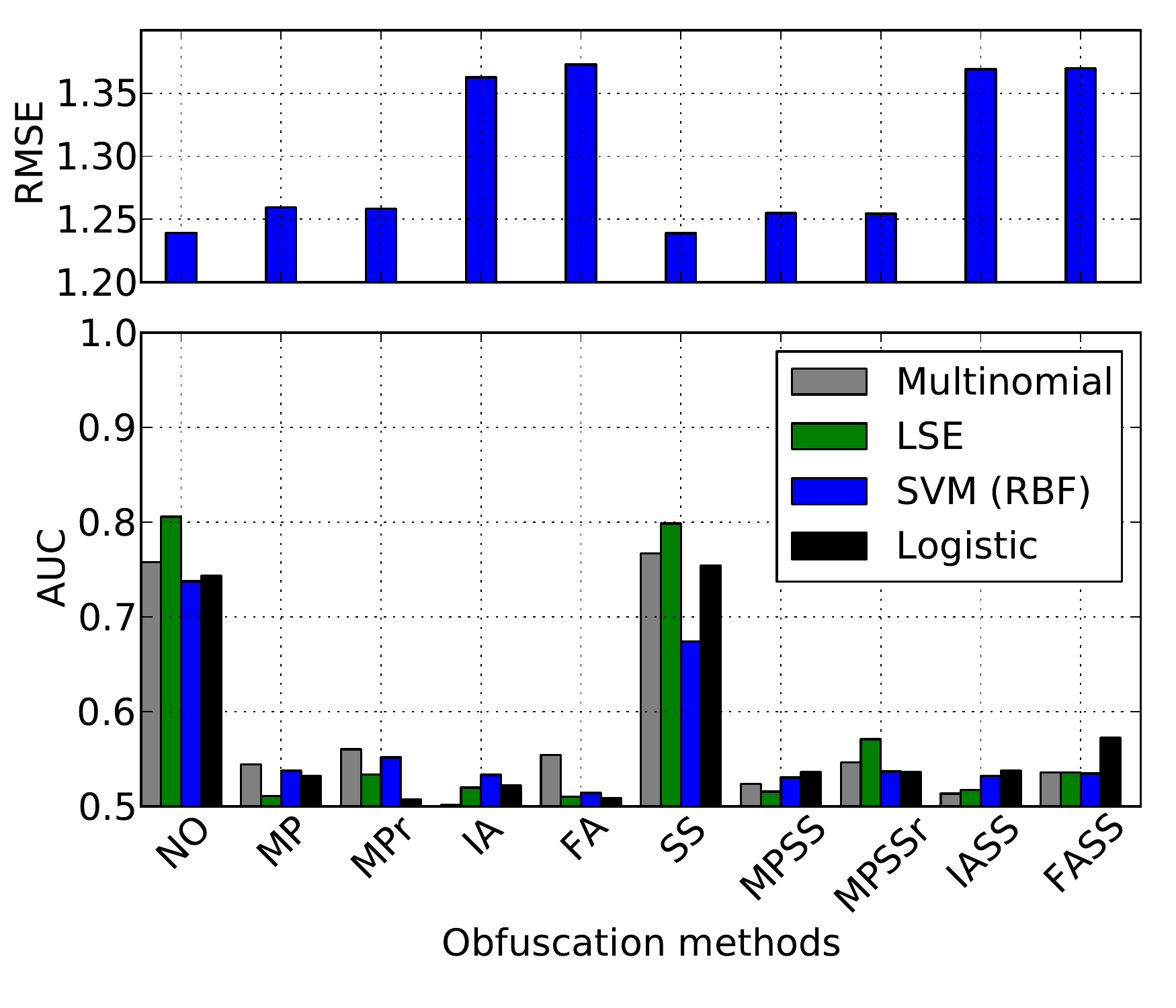}
\label{fig:ptv-politics}
}\caption{\small Privacy risk and prediction accuracy on PTV, obtained using four classifiers and obfuscation schemes (NO-No obfuscation, MP - Midpoint Protocol, r - Rounding, IA - Item Average, FA - Feature Average, SS - Sub-Sampling). The proposed protocol (MP) is robust to privacy attacks with hardly any loss in predictive power.}\label{fig:bar_plots}\end{minipage}\end{figure*}

\begin{figure*}[!t]
\centering
\subfloat[\ml~- Gender]{
\includegraphics[width=0.333333333333333333333\textwidth]{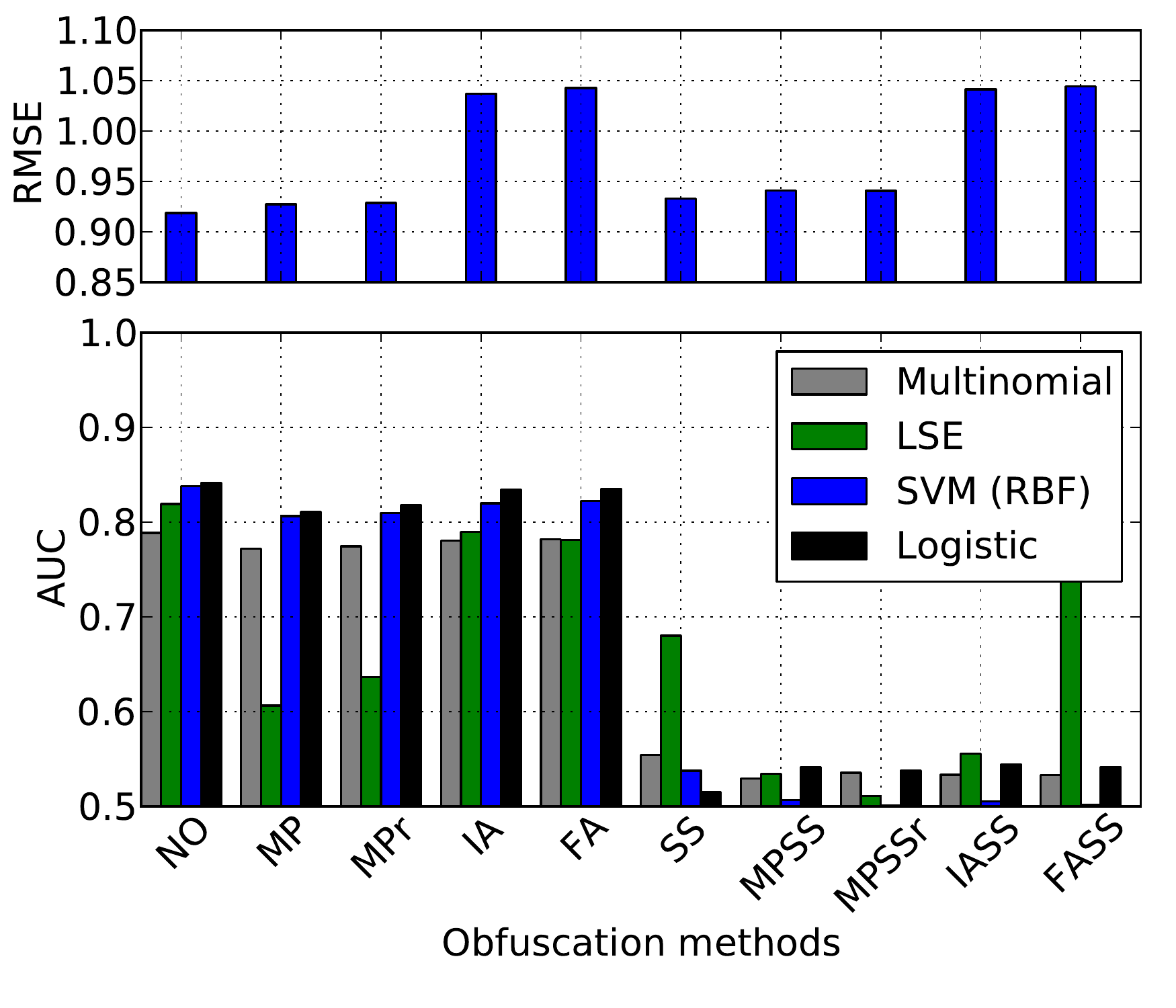}
\label{fig:ml-gender}
}
\subfloat[\ml~- Age]{
\includegraphics[width=0.333333333333333333333333\textwidth]{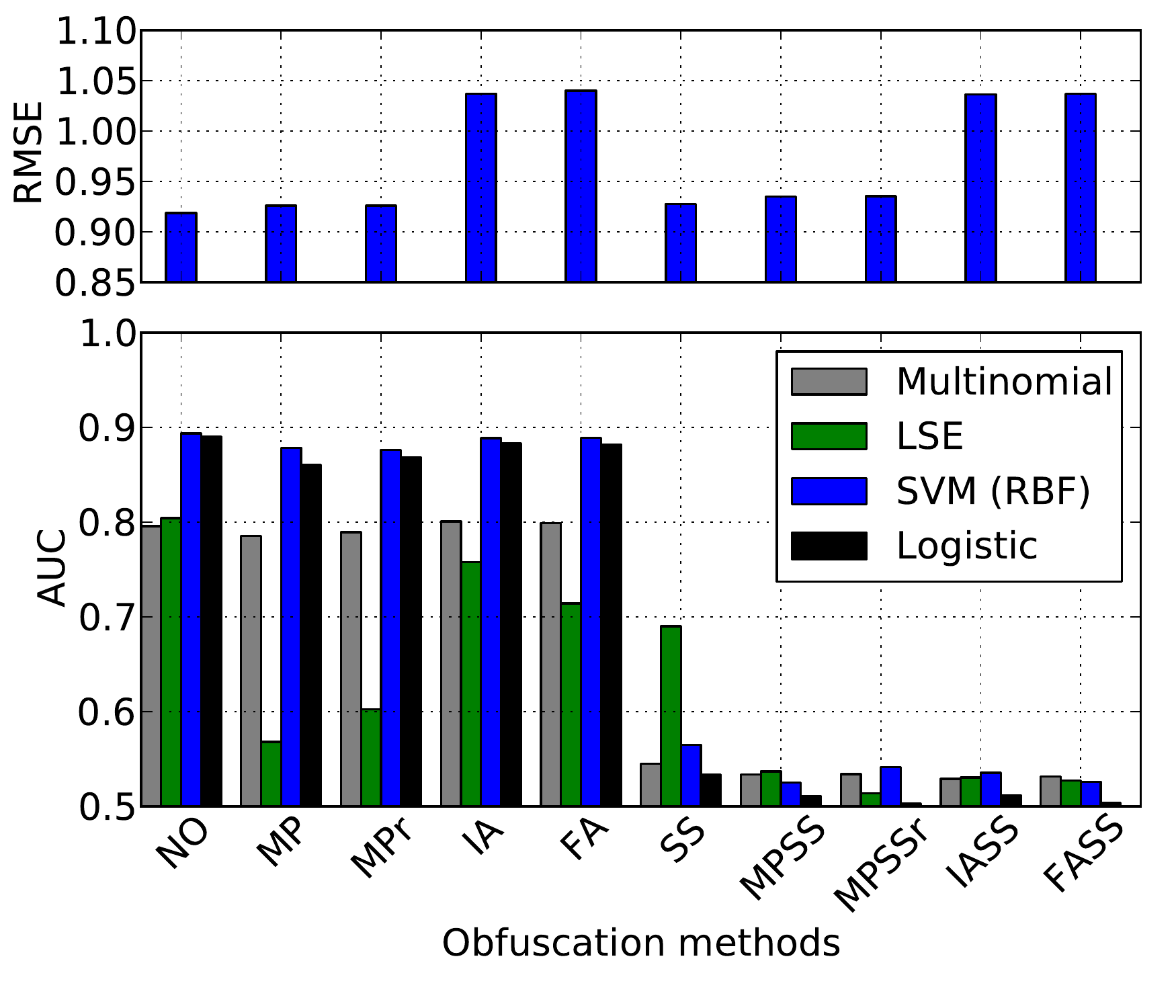}
\label{fig:ml-politics}
}
\subfloat[\fl~- Gender]{
\includegraphics[width=0.333333333333333333333333\textwidth]{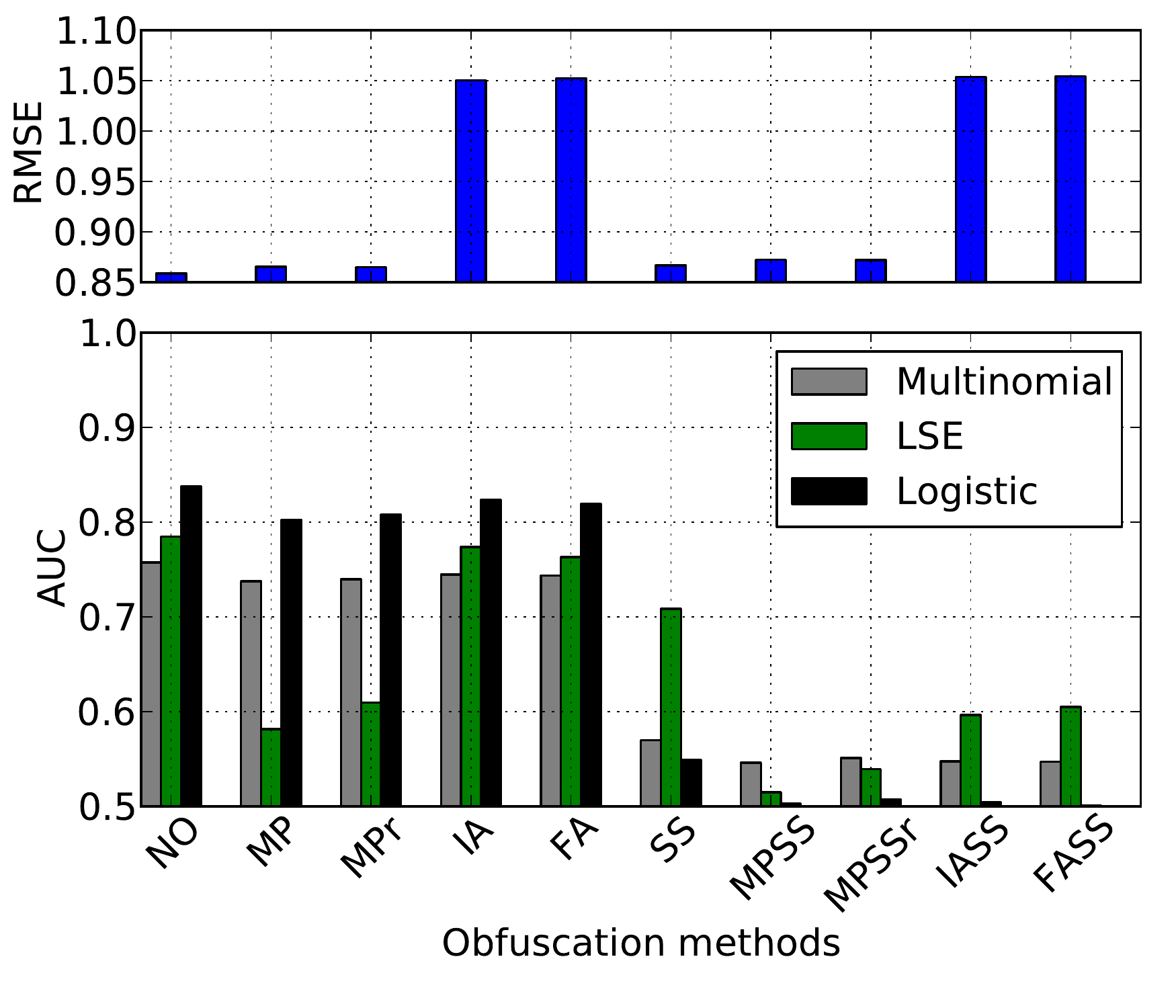}
\label{fig:fl-gender}
}
\caption{\small Privacy risk and prediction accuracy on \ml and \fl (sparse datasets), obtained using four classifiers and obfuscation schemes (NO-No obfuscation, MP - Midpoint Protocol, r - Rounding, IA - Item Average, FA - Feature Average, SS - Sub-Sampling). The proposed protocol (MPSS) is robust to privacy attacks without harming predictive power.}
\label{fig:bar_plots_ml_fl}
\end{figure*}

\section{Evaluation}
\label{sec:evaluation}

In this section we evaluate our protocols on real-world datasets. Our experiments
confirm that MP and MPSS are indeed able to protect the privacy of users
against inference algorithms, including non-linear algorithms, with little impact on prediction accuracy.

\subsection{Experimental Setup}
We study two types of datasets, sparse and dense. In sparse datasets, the set of items a user rates is often correlated with the private feature. This does not happen in dense datasets because all users rate all (or nearly all) items. 

\smallskip\noindent\textbf{Datasets.}
We evaluate our methods on three datasets: Politics-and-TV, \ml and \fl.
Politics-and-TV (PTV)~\cite{salman:2013} is a ratings dataset that includes
5-star ratings of users to 50 TV-shows and, in addition, each user's political affiliation (Democrat or
Republican) and gender. To make it dense, we consider only users that rate over 40 items, resulting in 365 users; 280 provide ratings to all 50 TV shows.
\ml\footnote{\url{http://www.grouplens.org/node/73}} and \fl\footnote{\url{http://www.sfu.ca/~sja25/datasets/}}~\cite{Jamali:2010} are movie recommender
systems in which users rate movies from a catalog of thousands of movies. Both \ml and \fl
datasets include user gender. \ml also  includes age groups; we categorize users
as \emph{young adults} (ages 18--35), or \emph{adults} (ages 35--65). We preprocessed \ml and \fl to consider only users with at least 20 ratings,
and items that were rated by at least 20 users.
Table~\ref{table:stats} summarizes the statistics of these three datasets.

\smallskip\noindent\textbf{Methodology.}
Throughout the evaluation, we seek to quantify the privacy risk to a user
as well as the impact of obfuscation on the prediction accuracy. To this end, we perform 10-fold cross-validation as follows. We split the users in each dataset into 10 folds. We use 9 folds as a \emph{training set} (serving the purpose of a dataset of non privacy-conscious users in Figure~\ref{fig:setting})  and 1 fold as a \emph{test set} (whose users we treat as privacy-conscious).

We use the training set to (a) compute extended profiles for each item by performing matrix factorization,  (b) empirically estimate the probabilities $p^+$, $p^-$ for each item, and (c)   train multiple classifiers, to be used to infer the private features. We describe the details of our MF implementation and the classifiers we use below.

We split the ratings of each user in the test set into two sets by randomly selecting 70\% of the ratings as the first set, and the remaining 30\% as the second set.  We obfuscate the ratings in the first set using MP, MPSS,  and several baselines as described in detail below. The obfuscated ratings are given as input to our classifiers to infer the user's private feature. We further estimate a user's extended profile using the LSE method described in Section~\ref{sec:newuser}, and use this profile (including both $\bx$ and the inferred $x_0$) to predict her ratings on the second set. For each obfuscation scheme and classification method, we measure the \emph{privacy risk} of the inference through these classifiers using the area under the curve (AUC)  metric \cite{hastie,hanley1982characteristic}. Moreover, for each obfuscation scheme, we measure the \emph{prediction accuracy} through the root mean square error (RMSE) of the predicted ratings. We cross-validate our results by computing the AUC and RMSE 10 times, each time with a different fold as a test set, and reporting average values.
We note that the AUC ranges from 0.5 (perfect privacy) to 1 (no privacy).

\smallskip\noindent\textbf{Matrix Factorization.} We use 20 iterations of stochastic gradient descent~\cite{Koren:2009} to perform MF on each training set. For each item, feature biases $v_{j0}$ were computed as the
half distance between the average item ratings per each private feature value. The
remaining features $\bv_j$ were computed through matrix factorization.
We computed optimal regularization parameters and the dimension $d=20$ through an additional 10-fold cross validation. 

\smallskip\noindent\textbf{Privacy Risk Assessment.}
We apply several standard classification methods to infer the private feature from the training ratings, namely Multinomial Na\"ive Bayes~\cite{blurme:2012}, Logistic Regression (LR), non-linear Support Vector Machines (SVM) with a Radial Basis Function (RBF) kernel, as well as the LSE \eqref{jointmle}. The input to the LR, NB and SVM methods comprises the ratings of all items provided by the user as well as zeros for movies not rated, while LSE operates only on the ratings that the user provides. As SVM scales quadraticaly with the number of users, we could not execute it on our largest dataset (\fl, c.f.~Table~\ref{table:stats}). 
\smallskip\noindent\textbf{Obfuscation Schemes.}
When using MP, the obfuscated rating may not be an integer value, and may even be
outside of the range of rating values which is expected by a recommender system. Therefore,
we consider a variation of MP that rounds the rating value to an integer in the range $[1,5]$.
Given a non-integer obfuscated rating $r$, which is between two
integers $k=\lfloor r \rfloor$ and $k+1$, we perform rounding by
assigning the rating $k$ with probability $r-k$ and the
rating $k+1$ with probability $1-(r-k)$, which on expectation gives the desired rating $r$.
Ratings higher than 5 and those lower than 1 are truncated to 5 and 1, respectively.
We refer to this process as {\em rounding}, and denote the obfuscation
scheme as MPr for midpoint protocol with rounding and MPSSr for
midpoint protocol with sub-sampling and rounding.

We also consider two alternative methods for obfuscation. First,
the {\em item average} (IA) scheme replaces a user's rating with the average rating of the item, computed from the training set.
Second, the {\em feature average} (FA) scheme replaces the user's rating with the
average rating provided by the feature classes (\eg, males and females), each with probability 0.5.

Finally, we evaluate each of the above obfuscation schemes, \ie, MP, MPr, IA and FA, together with
sub-sampling (SS). As a baseline, we also evaluated the privacy risk and the prediction accuracy when \emph{no obfuscation} scheme is used (NO).

\subsection{Experimental Results}
\smallskip\noindent\textbf{Dense Dataset.}
We begin by evaluating the obfuscation schemes on the dense PTV dataset using its two users' features (gender and political affiliation), illustrated in Figures~\ref{fig:ptv-gender} and \ref{fig:ptv-politics}, respectively. Each figure shows the privacy risk (AUC) computed using the 4 inference methods, and the prediction accuracy (RMSE) on applying different obfuscation schemes. 
Both figures clearly show that MP successfully mitigates the privacy risk (AUC is around 0.5) whereas the prediction accuracy is hardly impacted (2\% increase in RMSE). This illustrates that MP attains excellent privacy in practice, and that our modeling assumptions are reasonable: there is little correlation to the private feature after the category bias is removed. Indeed, strong correlations not captured by \eqref{eq:LinearModel} could manifest as failure to block inference after obfuscation, especially through the non-linear SVM classifier. This is clearly not the case, indicating that any dependence on the private feature not captured by \eqref{eq:LinearModel} is quite weak.

Adding rounding (MPr), which is essential for real-world deployment of MP, has very little effect
on both the AUC and RMSE.
Though IA and FA are successful in mitigating the privacy risk, they are suboptimal in terms of prediction. They severely impact the prediction accuracy, increasing the RMSE by roughly 9\%.
Finally, since this is a dense dataset, there is little correlation between the private feature and the set of items a user rates. Therefore, MP without SS suffices to mitigate the privacy risk.

\begin{figure}[t]
\centering
 \setlength{\unitlength}{0.63\columnwidth}
	\begin{picture}(0.7,0.56)
			\put(-0.1,-0.05){\includegraphics[width=0.55\columnwidth]{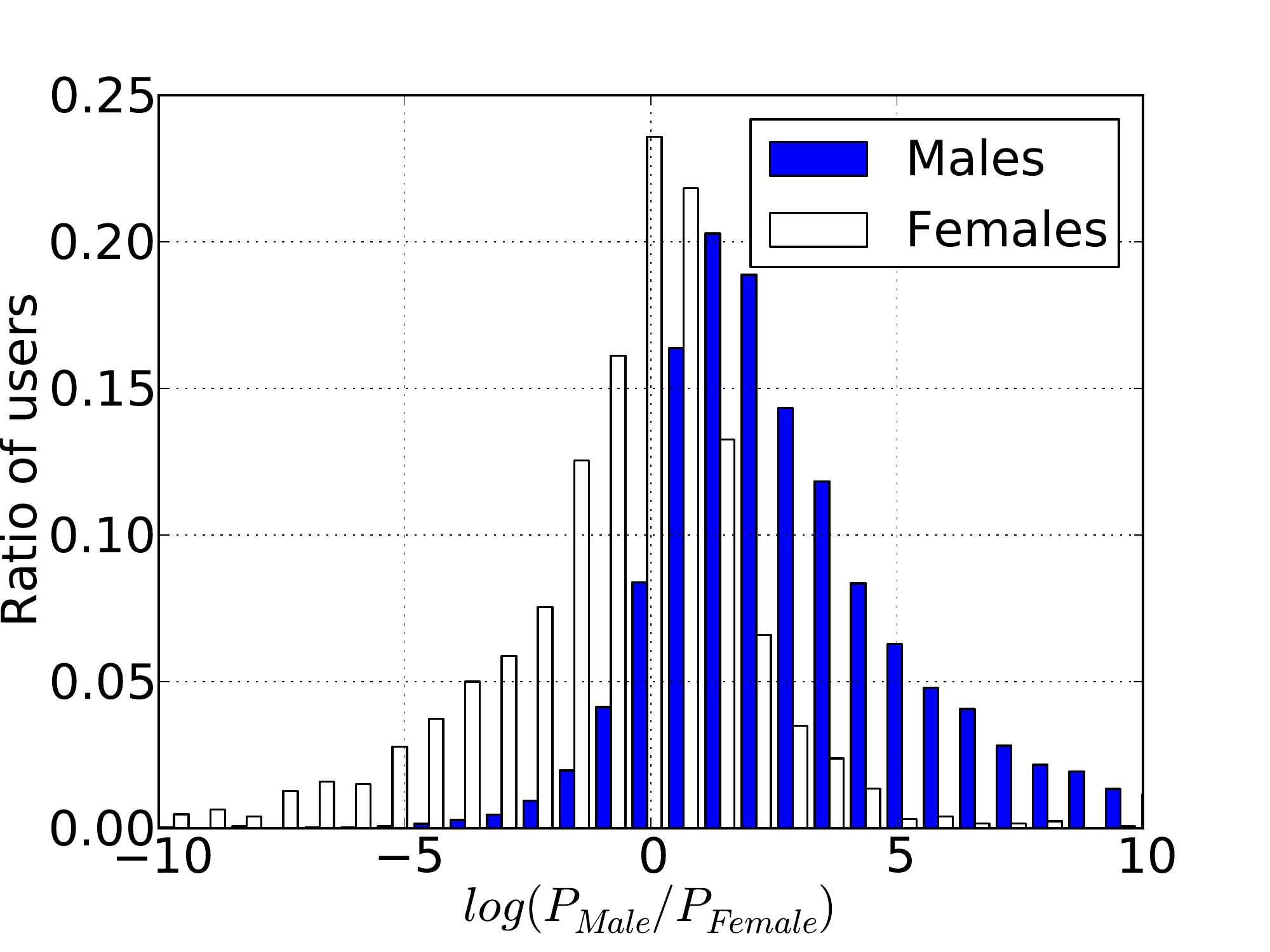}}
					\end{picture}	
          \begin{picture}(0.7,0.3)
			\put(0.0,-0.05){\includegraphics[width=0.55\columnwidth]{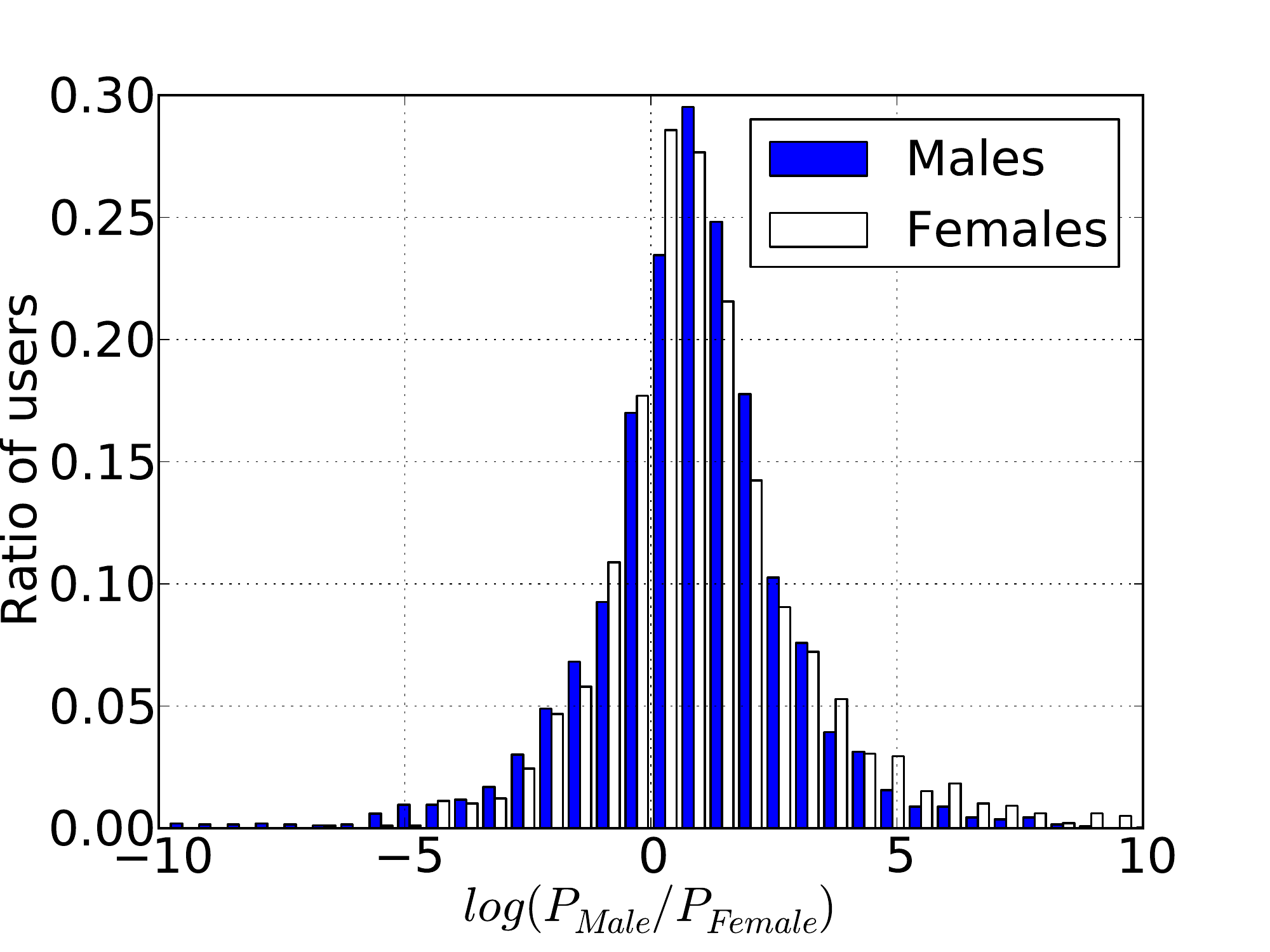}}
					\end{picture}
   \caption{\small Distribution of inference probabilities for males and females using \ml dataset and logistic regression (left) before obfuscation and (right) after MPSS obfuscation.}
    \label{fig:hists}
 \end{figure}

\smallskip\noindent\textbf{Sparse Datasets.}
Next, we investigate the effect of partial feedback by evaluating our obfuscation schemes on the \ml and \fl
datasets. In these datasets, in addition to the rating value, the set of items rated by a user can be correlated with her private feature. The results for obfuscation on \ml and \fl are in
Figure~\ref{fig:bar_plots_ml_fl}.

For all three datasets, we observe that MP successfully blocks inference  by LSE, but fails against the other three methods. This is expected, as the items rated are correlated to the private feature, and LSE is the only method among the four that is insensitive to this set. For the same reason, SS \emph{alone} defeats all methods \emph{except} LSE, which still detects the feature from the unobfuscated ratings (AUC 0.69--0.71).  Finally, MPSS and MPSSr have excellent performance across the board, both in terms of privacy risk (AUC 0.5--0.55) and impact on prediction accuracy (up to 5\%). In contrast, IA and FA significantly increase the RMSE (around 15\%).
We stress that, in these datasets, items are not rated independently as postulated by \eqref{product}. Nevertheless, the results above indicate that MPSS blocks inference in practice \emph{even when this assumption is relaxed.}

\ignore{
\begin{figure}[t]
\centering
\includegraphics[width=0.3\textwidth]{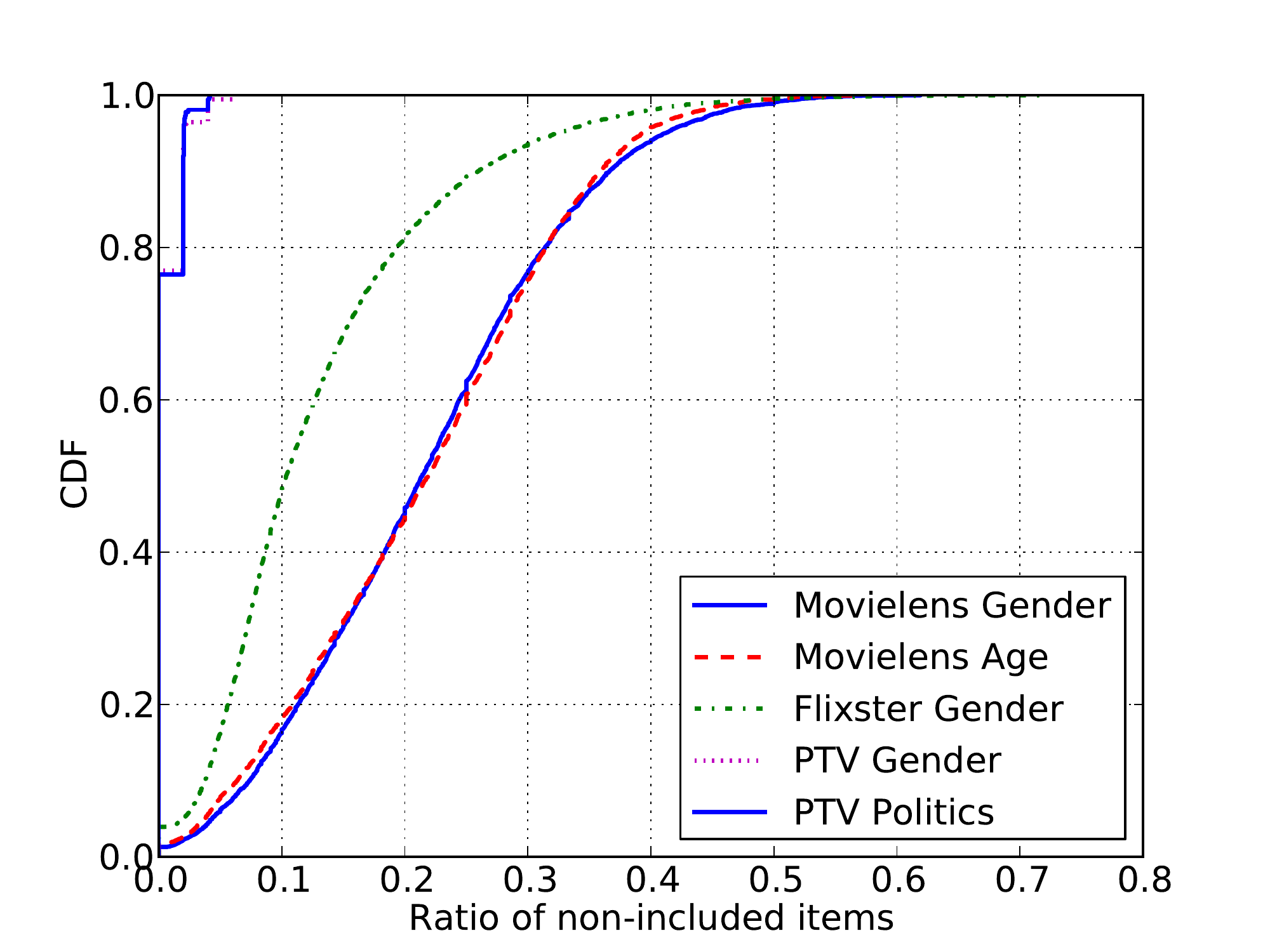}
\caption{\small Ratio of dropped items resulting from SS}
\label{fig:num_dropped}
\end{figure}
}

We further quantify the impact of sub-sampling in terms of the number of items that are not reported to the analyst in
a partial feedback setting. To this end, we compute the ratio of items excluded in the feedback reported by the user as a result of applying SS. We found that for the dense PTV dataset,  80\% of the users include all their ratings in their partial feedback, and the remaining 20\% exclude at most 5\% of their ratings.
For the sparse \fl and \ml datasets, 50\% of the users do not include 10\% and 23\% of their ratings, respectively.
All users include at least 50\% of their ratings, hence the prediction accuracy does not suffer with MPSS obfuscation. 

Overall, these results indicate that both MP and MPSS are highly effective in real-world datasets -- they mitigate the privacy risk while incurring very small impact on the prediction accuracy. Moreover, these obfuscation schemes work well even when facing non-linear inference methods, such as SVM, indicating that in practice, they eliminate any dependency between the ratings and the private feature.

\smallskip\noindent\textbf{Privacy Risk.}
To further illustrate how obfuscation defeats the inference of a private feature, we focus of the effect of obfuscation on logistic regression over the \ml Gender dataset. 
Figures~\ref{fig:hists}(a) and~\ref{fig:hists}(b) plot
the distribution of $\log\left(P_{Male}/P_{Female}\right)$
(a) before obfuscation
and (b) after obfuscation with MPSS. Here, $P_{Male}$ and $P_{Female}$ are the posterior probabilities for the two classes as obtained through
logistic regression.
Prior to obfuscation, there is
a clear separation between the distributions of males and females, enabling successful
gender inference (AUC 0.82 as shown in Figure~\ref{fig:ml-gender}). However, after obfuscation, the two distributions become
indistinguishable (AUC 0.54).

\smallskip\noindent\textbf{Privacy-Accuracy Tradeoff.}
Finally, we study the privacy-prediction accuracy tradeoff by applying an obfuscation scheme on an item rating with
probability $\alpha$, and releasing the real rating with probability $1-\alpha$. We vary the value
of $\alpha$ between 0 and 1 in steps of 0.1, that is,
when $\alpha=0$ no obfuscation is performed, and $\alpha=1$ means
that all ratings are obfuscated. For each $\alpha$, we measure the RMSE as well as the AUC of LSE. 

Figure~\ref{fig:auc-rmse} shows the resulting RMSE-AUC tradeoff curves
for MPSS, MPSSr and the two baseline obfuscation schemes with sub-sampling. The figure  shows that MPSS and MPSSr provide the best privacy-accuracy tradeoff (the slopes of the curves are almost flat), and consistently obtain
better prediction accuracy (lower RMSE) for the same privacy risk (inference AUC) than all other methods.

\begin{figure*}[t!]
\centering
\subfloat[\ml Gender]{
\includegraphics[width=0.2\textwidth]{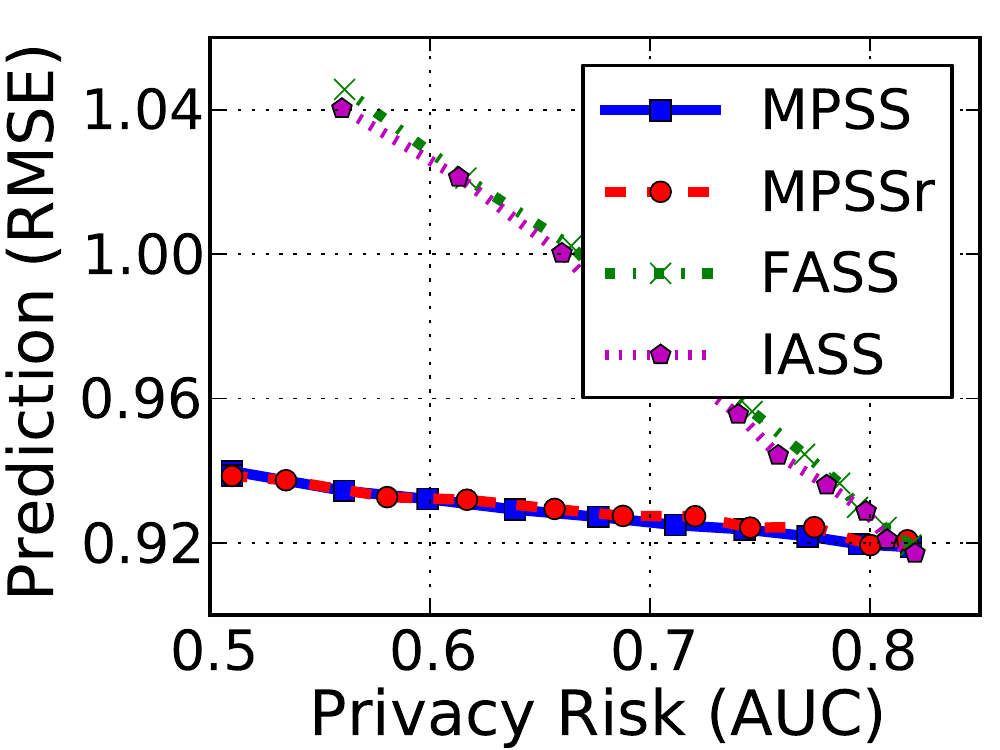}}\hspace*{-5pt}\subfloat[\ml Age]{
\includegraphics[width=0.2\textwidth]{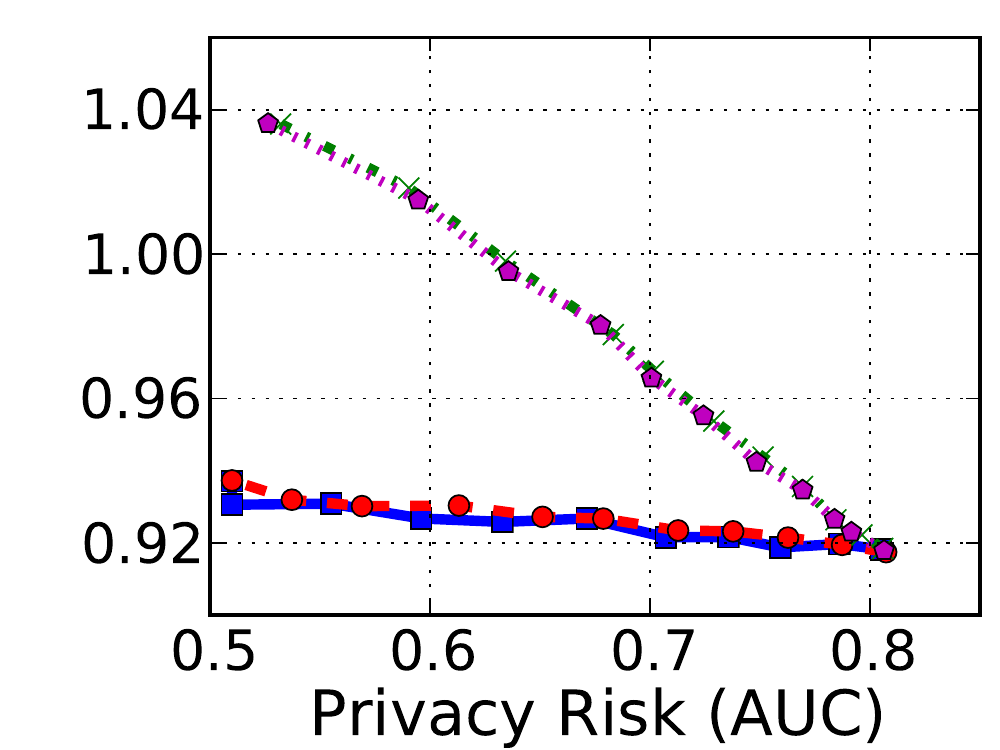}}\hspace*{-5pt}\subfloat[\fl Gender]{
\includegraphics[width=0.2\textwidth]{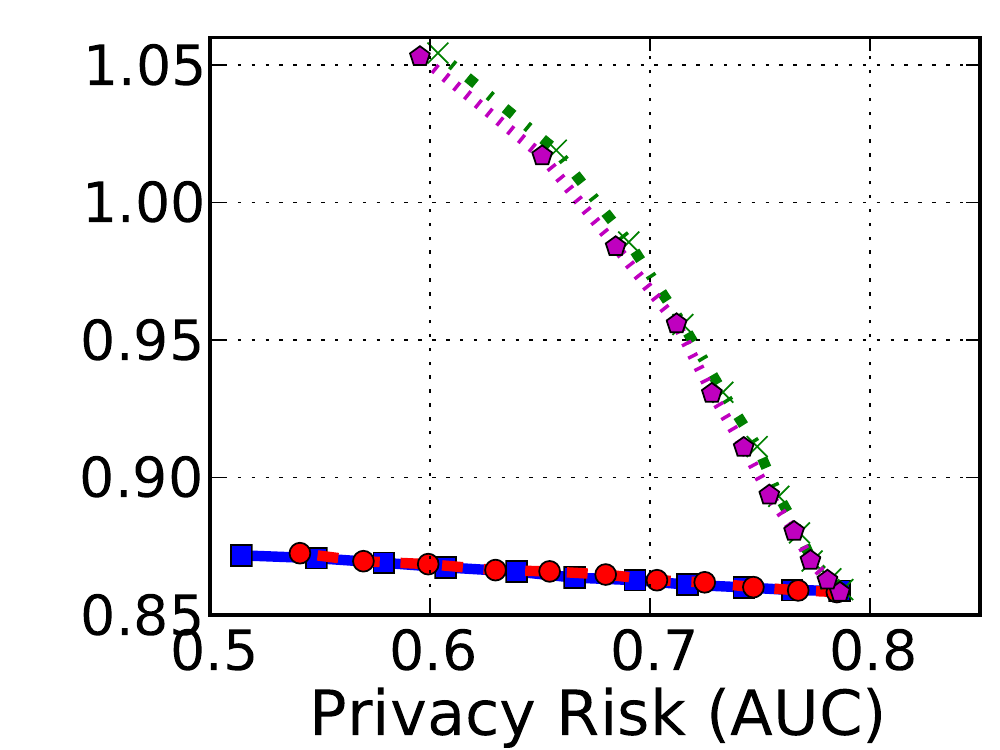}}\hspace*{\fill}\subfloat[PTV Gender]{
\includegraphics[width=0.2\textwidth]{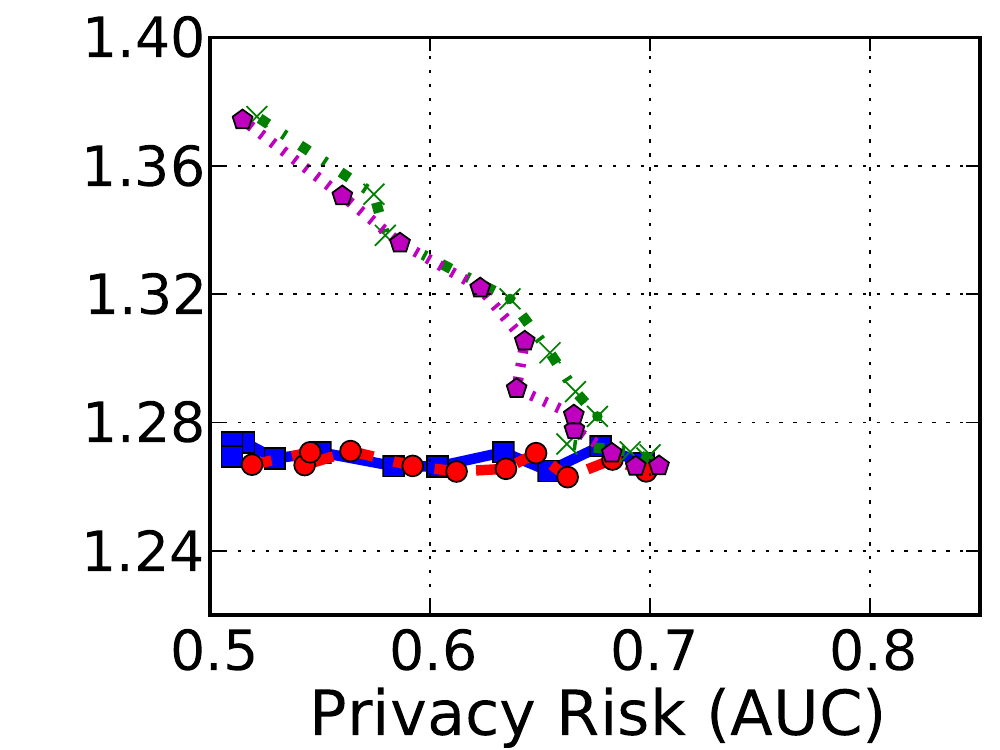}}\hspace*{-5pt}\subfloat[PTV Politics]{
\includegraphics[width=0.2\textwidth]{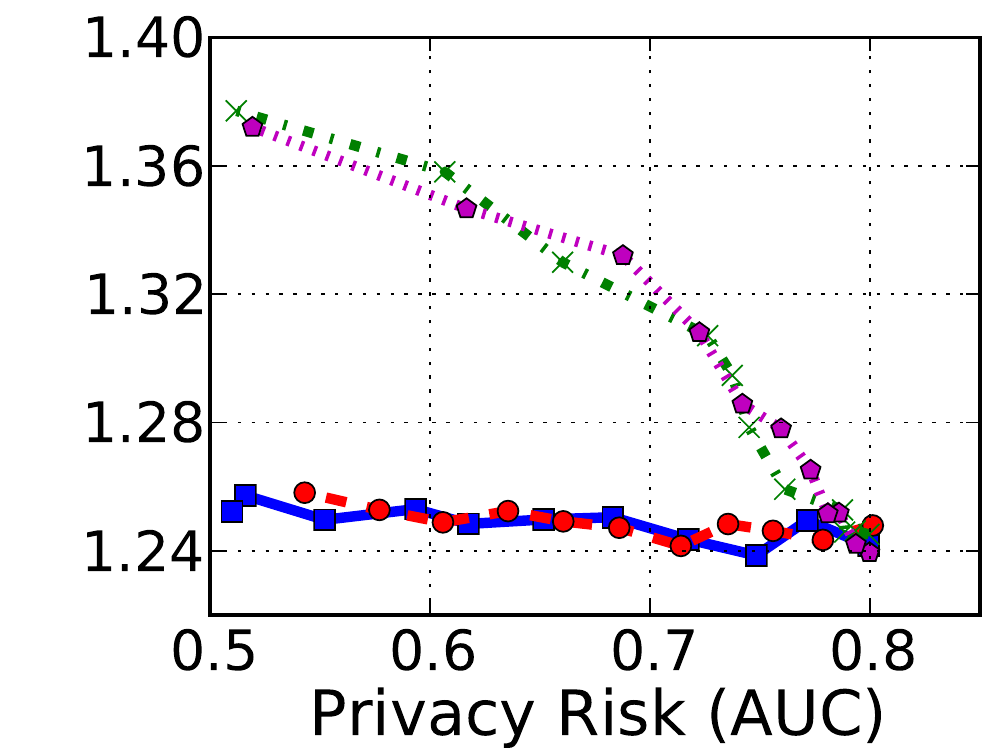}
}
\caption{\small Prediction accuracy (RMSE) vs.~privacy risk (LSE AUC) tradeoff for varying levels of obfuscation. Our proposed schemes (MPSS and MPSSr) have little impact on prediction accuracy as privacy is increased, whereas the prediction accuracy worsens dramatically under the baseline schemes.}
\label{fig:auc-rmse}
\end{figure*}

\makeatletter{}\section{Conclusion}
We have introduced a framework for reasoning about privacy, accuracy, and dislosure  tradeoffs in matrix factorization. This naturally raises the question of how these tradeoffs extend to other statistical or prediction tasks. An orthogonal direction to the one we pursued, when seeking a mininal disclosure, is to investigate schemes that are not perfectly private. It would be interesting to investigate, e.g., privacy-dislosure tradeoffs, rather than the usual privacy-accuracy tradeoffs one encounters in literature. For example, it is not clear whether one can construct protocols in which the distribution of the obfuscated output differs accross users with opposite private attribute by, e.g., an $\epsilon$ factor, but leak less information than MP: such protocols could, e.g., disclose a quantized version of the biases for each item.

\bibliographystyle{abbrv}
\begin{small}
\bibliography{references}
\end{small}

\appendix
\makeatletter{}\section{Optimality of MPSS}\label{app:proofofsub}
We begin by defining the class of learning protocols within which we will show that MPSS is optimal.

\medskip\noindent\textbf{Learning Protocols.} We  define a learning protocol as a tuple $\cR = (\leak,(\myset{S}_R,\obfs),\bhx)$ where:
\begin{packeditemize} 
\item The disclosure $L:\reals^{d+1}_{-\mathbf{0}}\times [0,1]\times [0,1]\to\mathcal{L}$ at each item $j\in \myset{S}$ is now a function of the extended profiles \emph{and} the rating probabilities, i.e., $\ell_j = L(v_j,p^+_j,p^-_j)$, $j\in \myset{S}$. We again denote by $\ell=L(\cV,p)\in \myset{L}^{|S|}$ the vector of disclosures.  
\item The obfuscated user feedback is constructed in two steps. First, the user computes a set $\myset{S}_R(S_0,x_0,\ell)\subseteq S_0$, which determines the items for which the she will reveal her rating to the analyst; second, having determined $\myset{S}_R$, the user produces an obfuscated output $y = \obfs(r_{S_R},x_0,\ell)$, where $r_{S_R}\in \reals^{|S_R|}$ the vector of ratings for items in set $\myset{S}_R$. 
Note that $\myset{S}_R$ is constrained to be a subset of $\myset{S}_0$: the user may only reveal ratings for a subset of the items she has truly rated. 
The feedback of the user to the analyst is the pair $(\myset{S}_R, y)$, i.e., the user reveals along with the feedback $y\in \cY$ the items for which she provides feedback.  Formally, these two are determined by a mapping $\myset{S}_R:2^{\myset{S}}\times\{-1,+1\}\times\cL^{|\myset{S}|}\to 2^{\myset{S}} $, that determines the set $\myset{S}_R\subseteq \myset{S}_0$, and a family of mappings  $\obfs:\reals^{|\myset{S}_R|}\times\{-1,+1\}\times\cL^{|\myset{S}_R|}\to \cY,$ one for each $\myset{S}_R\subseteq \myset{S}$. 
\item The estimator $\bhx=\bx((\myset{S}_R,y),\cV)$ is now a mapping $\bhx: 2^{\myset{S}}\times \cY \times (\reals^{d+1}_{-\mathbf{0}})^{|\myset{S}|}$ that depends on the user's feedback $(\myset{S}_R,y)$, as well as the profile information available to the analyst.
\end{packeditemize}
We can  naturally define partial orderings of learning protocols $\cR = (\leak,(\myset{S}_R,Y), \bhx)$  with respect to  the extent of disclosure by extending Definition \ref{def:disclosure} in a straightforward fashion. Regarding accuracy, we say that $\cR$ is more accurate than $\cR'$ if it yields a smaller expected $\ell_2$ loss \emph{conditioned on $\myset{S}_0$}. Finally, regarding privacy, we say that $\cR$ is privacy-preserving if the joint distribution of the random variables $(\myset{S}_R,y)$ does not depend on $x_0$: both the set $\myset{R}$, as well as the corresponding obfuscated feedback $y$, are equal in distribution when  $x_0=+1$ or $x_0=-1$ .   

We will further restrict our analysis to protocols that satisfy the following property.
\begin{definition}\label{def:pni}
Let 
\begin{align}\myset{S}^+ =\{j\in\myset{S}: \rho_j\leq 1\},\quad \myset{S}^- =\{j\in\myset{S}: \rho_j>1\}, \label{posneg}\end{align}
be the set of items more likely to be rated by ``positive'' and ``negative'' users respectively. 
We say that $\cR = (\leak, (\myset{S}_R,\obfs), \bhx)$ is \emph{positive-negative independent} (PNI) if the random sets
$\myset{S}_R\cap \myset{S}_+$ and $\myset{S}_R\cap \myset{S}_-$ are independent random variables.
\end{definition}
Note that MPSS is a PNI protocol, and so is any protocol in which the events $\{j\in\myset{S}_R\}$ are independent Bernoulli variables for every $j\in \myset{S}$.

\medskip\noindent\textbf{Optimality.} 
The following theorem holds
\begin{theorem}\label{thm:mpss}
Under \eqref{eq:LinearModel} with Gaussian noise, and \eqref{product}:
\begin{packedenumerate}
\item MPSS is privacy preserving.
\item There is no privacy preserving, PNI learning protocol that is strictly more accurate than MPSS.
\item  Any privacy preserving, PNI learning protocol that does not disclose as much information as MPSS must
also be strictly less accurate.
\end{packedenumerate}
\end{theorem}
\begin{proof}
We begin our proof of Theorem~\ref{thm:mpss} by establishing a few auxiliary results. Denote by $\cR=(\leak,(\myset{S}_R,\obfs),\bhx)$ be the MPSS protocol. Our first lemma states that \eqref{mpssprob} is an upper bound among privacy preserving protocols:
\begin{lemma}\label{lem:atmost} Let $\cR'=(\leak',(\myset{S}_R',\obfs'),\bhx')$ be privacy-preserving. Then
$\Pm_{x,\cV,p}(j \in \myset{S}_R') \leq \min (p_j^+,p_j^-) .$
 \end{lemma}
\begin{proof}
Recall that, by construction $\myset{S}_R'\subseteq \myset{S}_0$, the actual items rated by a user. Hence $\Pm_{(+1,\bx),cV,p}(j \in \myset{S}_R') \leq p_j^+$ and $\Pm_{(-1,\bx),cV,p}(j \in \myset{S}_R') \leq p_j^-$. As $\cR$ is privacy preserving, these l.h.s.~probabilities are equal, and the lemma follows.
\end{proof}
In fact, this inequality becomes strict if $\leak'$ does not disclose $\rho_j=p_j^-/p_j^+$, for some $j\in \myset{S}$. 

\begin{lemma}\label{lem:strict} Let $\cR'=(\leak',(\myset{S}_R',\obfs'),\bhx')$ be privacy-preserving, and suppose that $\leak'$ does not disclose $\rho={p^-}/{p^+}$--i.e., there is no $\phi:\cL'\to \reals$ such that $p^-/p^+ = \phi (\leak'(v,p^+,p^-))$ for all $v,p^+,p^-$. Then, there exist values $p^+,p^-\in [0,1]$,  an extended profile $v\in \reals^{d+1}_{-\mathbf{0}}$, and an $x_0\in\{-1,+1\}$ such that $\Pm_{x,\cV,p} (j\in \myset{S}_R)< \min (p_j^+,p_j^+)$ for all $\cV$ and $p$ such that $v_j=v$ and $(p_j^+,p_j^-)=(p^+,p^-)$.
\end{lemma}
\begin{proof}
Assumption that $\cR'$ does not disclose  $\rho_j$, for some $j\in \myset{S}$. Then, there exist probabilities $p^+, q^+, p^-,q^- \in [0,1]  $ and extended vectors $v,v'\in \reals^{d+1}_{-\mathbf{0}}$ such that $$\rho \equiv p^-/p^+ <  q^-/q^+\equiv\rho',$$ while $\leak(v,p^+,p^-)= \leak (v',q^+,q^-)$. 
Consider any two $\cV,\cV'\subseteq \reals^{d+1}_{-\mathbf{0}}$ and $p,p'\in ([0,1]\times [0,1])^{|\myset S|} $ such that all item profiles and probabilities are identical for all $k\myset{S}$, but differ in $j$: the j-th elements in $\cV,p$ are $v$ and $(p^+,p^-)$, respectively, while the $j$-th elements of $\cV',p'$ are $v'$ and $(q^+,q^-)$, respectively. Observe that $\equiv\leak'(\cV,p) =\leak(\cV',p')$.

Recall that $\myset{S}_R'$ depends on $\myset{S}_0$, $x_0$, and the disclosure from the analyst. Hence, as $\equiv\leak'(\cV,p) =\leak(\cV',p')$, conditioned on $\myset{S}_0$, $\myset{S}_R'$ is identically distributed in both cases. In particular,
\begin{align}\Pm_{x,\cV,p}(j\in \myset{S}_R' \mid \myset{S}_0=\myset{A} ) = \Pm_{x,\cV',p'}(j\in \myset{S}_R' \mid \myset{\myset{S}_0}=\myset{A}),\label{out}\end{align}  
for all 
$\myset{A}\subseteq \myset{S}$.
As $\myset{S}_R'\subseteq \myset{S}_0$, we have
$\Pm_{x,\cV,p}(j=\myset{S}_R')  \stackrel{\eqref{product}}{=} Z \cdot p_j^{x_0}$
where $$Z_{x,\cV,p} = \!\!\!\!\!\!\sum_{\myset{A}\subseteq \myset{S}\setminus\{j\}}\!\!\!\!\!\!\Pm_{x,\cV,p}(j\!\in\! \myset{S}_R'\! \mid\! \myset{S}_0=\myset{A}\cup\{j\} )\prod_{k\in \myset{A}}\!\!p_k^{x_0}\!\!\!\!\!\!\!\!\!\prod_{k\in \myset{S}\setminus (\myset{A}\!\cup\! \{j\})}\!\!\!\!\!\!\!\!\!(1-p_k^{x_0}) $$
As $\cR'$ is privacy preserving, by Lemma~\ref{lem:atmost} we get that $Z \leq \min(1,\rho^{x_0})$. Repeating the same steps for $\Pm_{x,\cV',p'}(j=\myset{S}_R')$,  we get that also $Z_{x,\cV',p'}\leq \min(1,(\rho')^{x_0}\}$. By \eqref{out}, these are equal, and thus
 $$Z=Z_{x,\cV,p} = Z_{x,\cV',p'}\leq \min(1,(\rho)^{x_0},(\rho')^{x_0})$$
Recall that $\rho<\rho'$, by construction. If $\rho<1$, then for $x_0=+1$ we get $\min(1,(\rho)^{x_0},(\rho')^{x_0})= \rho$. Then,
$$\Pm_{(+1,\bx),\cV',p'}(j=\myset{S}_R') = Z q^+\! \!\leq\! \rho q^+\!\! <\! \min(1,\rho') q^+\! \!=\! \min (q^+\!,q^-)$$
and the lemma holds for $x_0=+1$, $v'$, and $(q^+,q^-)$. If $\rho\geq 1$, then for $x_0={-1}$ we get $\min(1,(\rho)^{x_0},(\rho')^{x_0})= (\rho')^{-1}$, and
$$\Pm_{(-1,\bx),\cV,p}(j=\myset{S}_R') = Z p^- \!\!\leq \! p^-\!/\rho'\! <\! \min(1,\rho^{-1}\!) p^-\!\! = \!\min (p^+\!,p^-)$$
and the lemma holds for $x_0=-1$, $v$, and $(p^+,p^-)$.
\end{proof}
The PNI property allows us to couple $\myset{S}_R'$ and $\myset{S}_R$ in a way that the latter dominates the former.
\begin{lemma}\label{lem:dom}  Let $\cR=(\leak,(\myset{S}_R,\obfs),\bhx)$ be the MPSS protocol, and $\cR'=(\leak',(\myset{S}_R',\obfs'),\bhx')$ a privacy-preserving, PNI protocol. Then, there exists a joint probability space in which $\myset{S}_R' \subseteq \myset{S}_R \subset \myset{S}_0$.
\end{lemma}
\begin{proof}
Recall that $\myset{S}=\myset{S}^+\cup\myset{S}^-$, where $\myset{S}^+,\myset{S}^-$ the sets of positive and negative items in \eqref{posneg}.
Let $\myset{S}_{0+}$ and $\myset{S}_{0-}$ the set of items rated by two users of type $x_0=+1$ and $x_0=-1$, respectively.
We construct $\myset{S}_{0+}, \myset{S}_{0-}$ on the same probability space as follows:
for each $j$, draw $X_j$ uniform in [0,1], and let $j\in \myset{S}_{0+}$ iff
$X_j\le p_j^+$ and $j\in \myset{S}_{0-}$ iff $X_i\le p_i^-$. The sets  $\myset{S}_{0+}, \myset{S}_{0-}$
can  be intersected in the obvious way with M(+) and M(-) yielding
\begin{align*}\myset{S}_{0+} = \myset{S}_{0+}^+ \cup \myset{S}_{0+}^-, \quad
\myset{S}_{0-} = \myset{S}_{0-}^+ \cup \myset{S}_{0-}^-
\end{align*}
Then we have, a.s.
$\myset{S}_{0+}^+ \supseteq \myset{S}_{0-}^+$ and $\myset{S}_{0+}^- \subseteq   \myset{S}_{0-}^-$
Now, we can construct the set $\myset{S}_R$ reported by MPSS on the same space by letting
$\myset{S}_R = \myset{S}_R^+ \cup \myset{S}_R^-$ 
where
$\myset{S}_R^+ = \myset{S}_{0-}^+$ and  $\myset{S}_R^- = \myset{S}_{0+}^-$.

Now apply any privacy preserving mechanism $\cR'$ to $\myset{S}_{0+}$ and $\myset{S}_{0-}$. This will yield
sets $\myset{Q}_{R+}$, $\myset{Q}_{R-}$, that can also be decomposed as above:
$$\myset{Q}_{R+} = \myset{Q}_{R+}^+ \cup \myset{Q}_{R+}^-, \quad \myset{Q}_{R-} = \myset{Q}_{R-}^+ \cup \myset{Q}_{R-}^- $$
The sets $\myset{Q}_{R+}$, $\myset{Q}_{R-}$, are not necessarily equal, but must satisfy the following the properties:
\begin{align}
&\!\!\myset{Q}_{R-}^+\! \subseteq \!\myset{S}_{0-}^+\!\! =\! \myset{S}_R^+, ~   \myset{Q}_{R+}^- \!\subseteq\! \myset{S}_{0+}^-\!\! =\! \myset{S}_R^- ,~ \text{(by construction),}\label{bound1}\\
&\!\!\myset{Q}_{R-}^+ \!\ed\! \myset{Q}_{R+}^+,~~    \myset{Q}_{R-}^- \!\ed\! \myset{Q}_{R+}^- ,~ \text{(by privacy)} \label{bound2} 
\end{align}
where $\ed$ denotes equality in distribution.
Define
$\myset{S}_R' \equiv \myset{Q}_{R-}^+ \cup \myset{Q}_{R+}^-.$
By \eqref{bound1},  we get
$ \myset{S}_R'\subseteq \myset{S}_R$ with probability 1. Moreover, by  \eqref{bound2} and the fact that $\cR'$ is PNI, we get that
$\myset{S}_R' \ed \myset{Q}_{R+} \ed \myset{Q}_{R-}$.
\end{proof}
We are ready to prove Theorem~\ref{thm:mpss}. Privacy follows directly from \eqref{mpssprob}. Theorem~\ref{privacytheorem} implies MPSS yields minimal $\ell_2$ loss conditioned on $\myset{S}_R$. Optimality conditioned on $\myset{S}_0$ follows from Lemma~\ref{lem:dom}, and the fact that the $\ell_2$ loss \eqref{lsqloss} is a monotone decreasing function of $\myset{S}_R$. Finally, any protocol that does not does not disclose $v_{j0}$, for some $j\in \myset{S}$ will lead to a higher loss by Theorem~\ref{privacytheorem}. Moreover, by Lemmas~\ref{lem:strict} and~\ref{lem:dom}, if a protocol $\cR'$  does not disclose $\rho_j = {p_j^-}/{p_j^+}$ for some $j\in \myset{S}$, there exist $x_0$, $\cV$, and $p$ for which one can construct a coupling of $\myset{S}_R'$ and $\myset{S}_R$ such that $\myset{S}_R'\subset \myset{S}_R$ with non zero probability; the minimality of the disclosure therefore follows, again from the fact that the $\ell_2$ loss \eqref{lsqloss} is decreasing in $\myset{S}_R$.\end{proof}
 
\end{document}